%% file: main.tex
\documentclass{article}
\usepackage[utf8]{inputenc}

\usepackage[margin=1in]{geometry}

\usepackage[backref=page]{hyperref}
\usepackage{float,graphicx,verbatim,fullpage,amssymb,amsmath,amsthm,enumerate,multicol, xspace,xcolor,mathtools,thmtools,thm-restate,cleveref,xspace,tikz,caption,subcaption,wasysym,enumitem,ulem}
\usepackage{algorithm, cite}
\usepackage[noend]{algpseudocode}
\usepackage{enumitem}
\usepackage{comment}
\usepackage{mdframed}

\usetikzlibrary{calc}

\providecommand{\nnreals}{\mathbb{R}_{\geq 0}}

\providecommand{\reals}{\mathbb{R}}

\definecolor{darkpastelgreen}{rgb}{0.01, 0.75, 0.24}
\definecolor{bleudefrance}{rgb}{0.19, 0.55, 0.91}

\usepackage{mleftright}
\usepackage{hyperref}
    \hypersetup{ colorlinks=true, linkcolor=bleudefrance, filecolor=magenta, urlcolor=bleudefrance, citecolor=darkpastelgreen}

\usetikzlibrary{arrows.meta}
\tikzset{>={Latex[width=1.5mm,length=1.5mm]}}

\newfloat{procedure}{htbp}{loa}
\floatname{procedure}{Procedure}

\def\R{\mathbb{R}}

\newcommand{\cP}{\mathcal{P}}

\newcommand{\opt}{\textsf{OPT}}

\newcommand{\adv}{\textsf{ADV}}

\newtheorem{theorem}{Theorem}[section]

\newtheorem{lemma}[theorem]{Lemma}
\newtheorem{claim}[theorem]{Claim}
\newtheorem{corollary}[theorem]{Corollary}

\newtheorem{fact}[theorem]{Fact}

\theoremstyle{definition}
\newtheorem{definition}[theorem]{Definition}
\newtheorem{remark}[theorem]{Remark}

\input{newcommands}

\def\final{0}  
\def\iflong{\iffalse}
\ifnum\final=0  
\newcommand{\rnote}[1]{{\color{red}[{Raymond: \bf #1}]\marginpar{\color{red}*}}}
\newcommand{\enote}[1]{{\color{green}[{\small Elena: \bf #1}]\marginpar{\color{red}*}}}
\newcommand{\ynote}[1]{{\color{purple}[{Young-San: \bf #1}]\marginpar{\color{red}*}}}
\newcommand{\new}[1]{{\color{red} #1}}

\else 
\newcommand{\rnote}[1]{}
\newcommand{\enote}[1]{}
\newcommand{\ynote}[1]{}
\fi

\providecommand{\email}[1]{\href{mailto:#1}{\nolinkurl{#1}\xspace}}
\renewcommand{\emph}{\textit}

\title{Learning-Augmented Algorithms for Online Concave Packing and Convex Covering Problems}

 \author{
 Elena Grigorescu\thanks{Purdue University. Supported in part by NSF CCF-1910659, NSF CCF-1910411 and NSF CCF-2228814.
 E-mail: \email{elena-g@purdue.edu}}
 \and
 Young-San Lin\thanks{University of Melbourne. Work done while at Purdue University. Supported in part by NSF CCF-1910659, NSF CCF-1910411 and NSF CCF-2228814.
 E-mail: \email{nilnamuh@gmail.com}}
 \and
 Maoyuan Song\thanks{Purdue University.  Supported in part by NSF CCF-2127806 and NSF CCF-2228814. 
 E-mail: \email{song683@purdue.edu}}
 }

\date{\today}

\allowdisplaybreaks

\begin{document}

\maketitle

\begin{abstract}

\emph{Learning-augmented algorithms} have been extensively studied across the computer science community in the recent years, driven by advances in machine learning predictors, which can provide additional information to augment classical algorithms. Such predictions are especially powerful in the context of online problems, where decisions have to be made without knowledge of the future, and which traditionally exhibits impossibility results bounding the performance of any online algorithm. The study of learning-augmented algorithms thus aims to use external \emph{advice} prudently, to overcome classical impossibility results when the advice is accurate, and still perform comparably to the state-of-the-art online algorithms even when the advice is inaccurate.

In this paper, we present learning-augmented algorithmic frameworks for two fundamental optimizations settings, extending and generalizing prior works. For \emph{online packing with concave objectives}, we present a simple but overarching strategy that \emph{switches} between the advice and the state-of-the-art online algorithm. For \emph{online covering with convex objectives}, 
we greatly extend primal-dual methods for online convex covering programs 
by~\cite{azar2016online} (FOCS 2016) 
and previous learning-augmented framework for online covering linear programs from the literature, to many new applications.
We show that our algorithms break impossibility results when the advice is accurate, while maintaining comparable performance with state-of-the-art classical online algorithms even when the advice is erroneous. 

\end{abstract}

\input{intro}

\input{packing}

\input{convex}

\input{ell-q}

\input{applications}

\input{conclusion}

\newpage
\bibliographystyle{alpha}
\bibliography{reference}

\appendix

\input{convexvariant}

\end{document}

%% file: newcommands.tex
\usepackage{environ}

\newcommand{\eps}{\varepsilon}

\def\*#1{\mathbf{#1}}
\def\+#1{\mathcal{#1}}

\newcommand{\cO}{\mathcal O}

\def\bfone{\mathbf{1}}
\def\bfzero{\mathbf{0}}

\newcommand{\amax}{a_{\max}}
\newcommand{\amin}{a_{\min}}

\NewEnviron{problem}[1]{
	\begin{center}\fbox{\parbox{6in}{
				{\centering\scshape #1\par}
				\parskip=1ex
				\everypar{\hangindent=1em}
				\BODY
}}\end{center}}

\makeatletter
\newcommand*{\inlineequation}[2][]{
  \begingroup
    \refstepcounter{equation}
    \ifx\\#1\\
    \else
      \label{#1}
    \fi
    \relpenalty=10000 
    \binoppenalty=10000 
    \ensuremath{
      #2
    }
    ~\@eqnnum
  \endgroup
}
\makeatother

\makeatletter
\renewcommand*{\@fnsymbol}[1]{\textcolor{darkpastelgreen}{\ensuremath{\ifcase#1\or *\or \dagger\or \ddagger\or
 \mathsection\or \triangledown\or \mathparagraph\or \|\or **\or \dagger\dagger
   \or \ddagger\ddagger \else\@ctrerr\fi}}}
\makeatother

%% file: intro.tex
\section{Introduction}

In the classical online model, algorithms are given parts of the input over time, and must make irrevocable decisions to process the input elements as they arrive, without knowledge of the future. 
The performance of an online algorithm is often measured by its ``competitive ratio", which is defined as the ratio between the ``cost" or ``value" of the algorithm's solution and that of the optimal solution in hindsight, or equivalently that of the solution of an optimal offline algorithm.
Due to the uncertainty around future inputs, online problems are traditionally hard, and often exhibit impossibility results, which are lower bounds on the competitive ratio of any online algorithm (e.g., \cite{kmmo,alon2009online,balseiro2019learning}).
We refer the readers to~\cite{buchbinder2009survey} and~\cite{hoi2021online} for excellent surveys on online algorithms. 

The on-going success of machine learning has led to the break-through concept of using predictions~\cite{purohit2018improving,lykouris2021competitive}, i.e., additional information inferred from historical data about future inputs or the problem instance as a whole, and incorporating them into classical online algorithms, to achieve better performance and break impossibility results. However, due to the probabilistic nature of machine learning, as well as the inability of online algorithms to verify the accuracy of these predictions, blindly trusting and following the advice can lead to undesirable performance compared to even online algorithms without predictions, as established by~\cite{szegedy2013intriguing} and~\cite{bamas2020primal}. As a result, the focus of the community has shifted to \emph{learning-augmented algorithms}, namely algorithms that utilize the advice prudently, maintaining rigorous guarantees and high performance, both when the advice is accurate, in which case we say it maintains \emph{consistency}, and when the advice is arbitrarily inaccurate, in which case we say that it satisfies \emph{robustness}. We refer the readers to~\cite{mitzenmacher2020algorithms} for a survey on learning-augmented algorithms.

\paragraph{Online covering and packing problems.} A recent line of work studies learning-augmented algorithm in context of online covering linear programs (LPs). The seminal work of~\cite{bamas2020primal} built upon the primal-dual framework of~\cite{buchbinder2009online} and presented the first primal-dual learning-augmented (PDLA) algorithm for online set cover. 
Combining their algorithms with ideas from~\cite{elad2016online}, the work of~\cite{grigorescu2022learning} generalized the primal-dual learning-augmented framework to solve general online covering LPs and semidefinite programs, allowing for fractional cost, constraint, and advice.
The PDLA framework utilizes the duality of covering and packing linear programs, and maintains a primal covering solution as well as a dual packing solution simultaneously, while fine-tuning the growth rate of each individual variable using the advice provided to the algorithm. Compared to many other learning-augmented algorithms for specific online problems, the PDLA framework is general-purpose, and can apply to a variety of problems oblivious of structure.

\paragraph{Our contributions at a glance.} 
In this paper, we extend the line of work on online covering and packing problems, by designing algorithmic frameworks for online packing maximization with concave objectives, and online covering minimization with convex objectives, further generalizing the setting of~\cite{grigorescu2022learning}. 
Our settings and frameworks are general-purpose, and can be directly applied to a variety of problems, agnostic of problem structure.
Concave packing and convex covering are traditionally considered strongly associated ``dual'' problems to each other in operations research and online optimization; Nonetheless, we show that these two settings admit vastly different learning-augmented algorithms despite their similarity. 

For online packing with concave objectives, we present a simple class of algorithms utilizing a similar ``switching" strategy to that outlined in~\cite{grigorescu2022arxiv} (attributed to Roie Levin), complementing their observation that even such simple strategies can outperform sophisticated algorithms for online covering linear programs. Switching between different algorithms has been a classical design philosophy in classical online algorithms~\cite{fiat1991competitive}, but is much less prevalent in learning-augmented algorithms (e.g., \cite{antoniadis2023online}).

For online covering with convex objectives, it appears that ``switching" strategies do not admit efficient algorithms, thus we instead present a primal-dual-based framework, taking inspiration from the classical primal-dual method for online convex covering in~\cite{azar2016online}. As an extended result, we adapt our PDLA framework to the online covering setting with $\ell_q$-norm objectives studied in~\cite{shen2020online}, utilizing problem structure to obtain better analyses. 
Finally, we apply our proposed frameworks to a variety of online optimization problems, including network utility maximization, optimization under inventory constraints, mixed covering and packing, and buy-at-bulk network design.

Conceptually, from our observations around the ``switching" paradigm, we raise some basic questions: What properties enables these switching strategies to perform well, that is present in concave packing, but absent in convex covering? More generally, which online problems allow such simple algorithms and solutions to exist? We believe that understanding these conceptual questions are important for the field of learning-augmented algorithms.

\paragraph{Organization.}  
We state necessary background knowledge in \Cref{sec:prelim}, and outline our contributions in more detail in \Cref{sec:contribution}. We survey and provide additional references to prior works in \Cref{sec:Prior}. In \Cref{sec:Switch}, we present a simple ``switching"-based overarching framework for online concave packing. In \Cref{sec:Convex}, we present an overview for our primal-dual learning-augmented framework for online convex covering. In \Cref{sec:Lq}, we present our extension of our PDLA framework onto online covering with $\ell_q$-norm objectives. In \Cref{sec:Appl}, we apply our algorithmic frameworks onto a variety of well-motivated practical problems. We conclude our paper with closing remarks and discussions of future directions in \Cref{sec:Conclusion}.

\subsection{Preliminaries} \label{sec:prelim}

We denote by $\bfzero_{(n)}$ and $\bfone_{(n)}$ the vector of all zeroes and all ones of length $n$, respectively. When the length of the vector is unambiguous, we omit the subscript and use $\bfzero$ and $\bfone$. We use $x \geq y$ for $x,y \in \nnreals^n$ to denote the relation that $x_i \geq y_i$ for all $i \in [n]$. 
A function $f : \reals^n \to \reals$ is monotone if for all $x \geq x'$, $f(x) \geq f(x')$ as well\footnote{Such functions are usually called `monotone non-decreasing'. Since all monotone functions in this paper are non-decreasing, we omit the classifier for simplicity.}.

We use $A := \{a_{ij}\}_{i \in [m], j \in [n]} \in \R_{\geq 0}^{m \times n}$ to denote the constraint matrix of the covering and packing problems we study, and we use $i \in [m]$ and $j \in [n]$ to denote the row and column index of $A$, respectively. 

\paragraph{Online concave packing.} The online packing setting with concave objectives we study is defined as follows:
\begin{align}
  \begin{aligned}
    \text{max } & g(y) 
    \text{ over } y \in \nnreals^m 
    \text{ subject to } A^T y \leq b.
  \end{aligned} \label{eq:concavepacking}
\end{align}
Here, $g: \nnreals^m \to \nnreals$ is a concave and monotone objective function with $g(\mathbf{0}) = 0$, and $b \in \R_{> 0}^n$ denotes an upper bound vector for the linear constraints. 

In the online setting, the objective function $g$, the values in $b$, and the number of constraints $n$ are given in advance. In each round $i \in [m]$, a new packing variable $y_i$ is introduced, along with all the associated coefficients $a_{ij}$ for all $j \in [n]$ (the $i$-th row of $A$). The number of packing variables $m$ might be unknown, and each packing constraint is gradually revealed column-wise to the algorithm. In round $i \in [m]$, the algorithm can only irrevocably assign a value for $y_i$. The goal is to approximately maximize $g(y)$ by assigning values to the variables online while maintaining (approximate) feasibility.

\paragraph{Online convex covering.} The online covering setting with convex objectives we study is defined as follows:
\begin{align}
  \begin{aligned}
    \text{min } & f(x) 
    \text{ over } x \in \nnreals^n 
    \text{ subject to } A x \geq \bfone.
  \end{aligned} \label{eq:convexcovering}
\end{align}
Here, $f: \nnreals^n \to \nnreals$ is a convex, monotone, and \emph{differentiable} function, with $f(\bfzero) = 0$. We make the technical assumption that the gradient of $f$, $\nabla f$, is monotone as well, and later show that it is possible to remove this assumption in a more structured setting.

The online setting for convex covering is similar to that of concave packing: The objective function $f$ and the numbers of covering variables $n$ are given in advance, but the constraint matrix $A$ and the number of constraints $m$ is unknown to the algorithm. In each round $i \in [m]$, the $i$-th row of $A$ arrives online, and the algorithm must update $x_j$ for all $j \in [n]$ in a non-decreasing fashion to satisfy the constraint, while approximately minimizing $f(x)$.

\paragraph{Online optimization with advice.} In the learning-augmented setting, the algorithm is additionally given an advice on the variables of interest: $y' \in \nnreals^m$ for packing, and $x' \in \nnreals^n$ for covering. We do not make any additional assumptions about the advice, and their objective values $g(y')$ or $f(x')$ can be arbitrarily worse compared to the optimal solution. Our frameworks additionally utilize a hyper-parameter $\lambda \in [0, 1]$, denoted the \emph{confidence parameter}, chosen by the user a priori.

The advice can be interpreted as a suggestion on what the solution, and thus what the main variables at the end of the algorithm's execution should be, and the confidence parameter $\lambda$ indicates how much the user trusts this advice. A value of $\lambda$ close to $0$ indicates that the advice is trusted and thus likely to be accurate, while a value close to $1$ indicates that the advice should not be trusted. The algorithm's goal is thus to incorporate the advice prudently, obtaining a final solution $\bar x$ or $\bar y$ whose performance is comparable to the advice when $\lambda$ is small, and comparable to the state-of-the-art online algorithm when $\lambda$ is large.

We formalize the metrics we use to measure the performance of learning-augmented algorithms via two notions, \emph{consistency} and \emph{robustness}:
\begin{definition}
An online (covering) solution $\bar x$ is $C(\lambda)$-consistent if $f(\bar x) \leq C(\lambda) \cdot f(x')$, where $f(x')$ is the cost of the advice. A learning-augmented algorithm is $C(\lambda)$-consistent if the solution it generates is $C(\lambda)$-consistent.
\end{definition}
\begin{definition}
An online (covering) solution $\bar x$ is $R(\lambda)$-robust if $f(\bar x) \leq R(\lambda) \cdot \opt$, where $\opt$ is the cost of the optimal offline solution. A learning-augmented algorithm is $R(\lambda)$-robust if the solution it generates is $R(\lambda)$-robust.
\end{definition}

The packing version of these definitions follows symmetrically: A solution $\bar y$ is $C$-consistent and $R$-robust if $g(\bar y) \geq \frac{1}{C} g(y')$ and $g(\bar y) \geq \frac{1}{R} \opt$. For packing maximization problems, it is common practice to allow the solutions found by an online algorithm to violate the packing constraints by a certain factor; an exactly feasible solution can be recovered by scaling down the approximately feasible solution. We make this notion of approximate feasibility explicit in the following definition:
\begin{definition}
An online (packing) solution $\bar y$ is $V(\lambda)$-feasible if $A^T \bar y \leq V(\lambda) \cdot b$. An online algorithm is $V(\lambda)$-feasible if the solution it generates is $V(\lambda)$-feasible.
\end{definition}

Intuitively, when we are confident in the advice, we should follow it as much as possible and obtain a solution that is close to the advice, so $C(\lambda)$ should tend to 1 as $\lambda$ tends to 0. On the other hand, when the advice is possibly inaccurate, our algorithm should follow the non-learning-augmented classical online algorithm, so $R(\lambda)$ should tend to the competitive ratio of the state-of-the-art online algorithm as $\lambda$ tends to 1.
The role of the confidence hyper-parameter $\lambda$ is thus to control the \emph{tradeoff} between consistency and robustness.

\subsection{Our Contributions} \label{sec:contribution}

\subsubsection{Online Packing with Concave Objectives}

Our result for the online packing problem with concave objectives is a general framework for devising learning-augmented algorithms for all online concave packing problems. 
We present a simple algorithm utilizing a switching strategy, which is reminiscent of the switching algorithm mentioned in~\cite{grigorescu2022arxiv}, but ultimately relies on properties of maximization problems distinct to that for covering minimization problems.
The algorithm uses any state-of-the-art classical online algorithm as a black-box subroutine, and obtains a solution that matches both the value of the advice and the value of the subroutine online algorithm asymptotically:
\begin{theorem}[Informal]\label{thm:packinginformal}
Given an instance of the online concave packing problem \eqref{eq:concavepacking}, there exists an online algorithm that takes an $\alpha$-competitive $\beta$-feasible online algorithm as a subroutine, an advice $y' \in \nnreals^n$ which is $\beta$-feasible, and a confidence parameter $\lambda$. The algorithm is $\frac{1}{1-\lambda}$-consistent, $\frac{\alpha}{\lambda}$-robust, and $(2 - \lambda) \beta$-feasible.
\end{theorem}

The formal version of \Cref{thm:packinginformal} is \Cref{thm:packingmain}. We state and analyze our algorithm in \Cref{sec:Switch}.

We remark that our algorithm is a general-purpose framework, and does not rely on any specific classical online algorithm. This property absolves users of our framework from the responsibility of understanding potentially sophisticated online algorithm used as a black-box subroutine. In addition, future advancements in the field of classical online problems would immediately imply advances in their variants augmented by advice. 

Our advice model on the packing variables also generalizes many prior models (e.g.,\cite{IKMP2021} for online knapsack problems), since any form of advice that suggests a course of action can be simulated by our packing program formulation by setting the advice entries to an appropriate value. The exact values of the advice entries depend on the exact problem and advice formulation: We give some examples of these reductions in \Cref{sec:ApplPacking}.

\subsubsection{Online Covering with Convex Objectives}

As our main result, we present a general framework for designing learning-augmented algorithms for online covering problems with convex objective functions \eqref{eq:convexcovering} and possibly non-integral advice and constraints, which generalizes previous works such as online set cover~\cite{bamas2020primal}, and online general covering LPs~\cite{grigorescu2022learning}.
We employ the standard assumption that the objective function $f$ is convex, monotone, differentiable, with $f(\bfzero) = 0$. We additionally assume that the gradient $\nabla f$ is monotone, which is a technical assumption made in works studying the classical online version of the problem as well~\cite{buchbinder2014online,azar2016online}.

We point out that the switching strategy for online covering linear programs in~\cite{grigorescu2022arxiv} crucially relies on the linearity (more specifically, subadditivity) of the objective function, which is not satisfied by convex functions in general, even with the monotone gradient assumption. The cost of ``switching" between solutions is at most a multiplicative factor of 2 in the linear case, but may be unbounded for convex objectives. Thus, such simple switching strategies do not apply to our setting.

Motivated by this, our algorithm for online convex packing returns to more sophisticated methods and follows the primal-dual learning-augmented (PDLA) paradigm~\cite{bamas2020primal,grigorescu2022learning}, taking the advice into consideration while carefully tuning the growth rate of each variable, in order to approximately minimize the cost of the solution. The PDLA algorithms we present are both robust and consistent, with performance close to that of the optimal offline solution when the advice is accurate, and close to that of the state-of-the-art online algorithm when the advice is inaccurate. The performance of our algorithm is (informally) characterized by the following theorem:
\begin{theorem}[Informal]\label{thm:convexinformal}
Given an instance of the online convex covering problem \eqref{eq:convexcovering}, there exists an online algorithm that takes an advice $x' \in \nnreals^n$ and a confidence parameter $\lambda$. The algorithm is $O(\frac{1}{1-\lambda})$-consistent and $O((p \log \frac{d}{\lambda})^p)$-robust. Here, $p := \sup_{x \geq \bfzero} \frac{\langle x, \nabla f(x) \rangle}{f(x)}$, and $d$ is the row sparsity of the constraint matrix.
\end{theorem}
The formal version of \Cref{thm:convexinformal} is \Cref{thm:convexmain} in \Cref{sec:Convex}, which addresses the case when the advice $x'$ is infeasible. We remark that our consistency ratio tends to 1 as $\lambda$ tends to $0$, and our robustness ratio tends to $O((p \log d)^p)$ as $\lambda$ tends to $1$, matching the competitiveness of the online algorithm presented in~\cite{azar2016online}, meeting the ideal expectations. While our algorithms and analyses resemble their counterparts in prior works, we employ some subtle yet vital modifications, since a more direct application of~\cite{grigorescu2022learning}'s learning-augmentation model onto~\cite{azar2016online}'s primal-dual framework suffers from various technical issues and cannot yield satisfactory results.

\subsubsection{Online Covering with $\ell_q$-norm Objectives}

Our technical assumption that the gradient is monotone follows from an identical assumption in~\cite{buchbinder2014online} and~\cite{azar2016online}. Subsequently, \cite{shen2020online} restricted their attention to online covering problems with $\ell_q$-norm objectives, and presented new analyses of a variant of~\cite{azar2016online}'s algorithm that removed the monotone gradient assumption,  using the structural properties of $\ell_q$-norms. Specifically, they consider covering problems of the following form:
\begin{equation}\label{eq:lqcovering}
\text{min } \sum_{e=1}^r c_e || x (S_e) ||_{q_e} \text{ over } x \in \nnreals^n \text{ subject to } A x \geq 1
\end{equation}
where each $S_e \subseteq [n]$ is a subset of indices, $q_e \geq 1$, $c_e \geq 0$, $x(S)$ denote the vector $x$ restricted to indices in $S$, and $||x(S)||_q$ is the $\ell_q$-norm, i.e., $\left( \sum_{i \in S} x_i^q \right)^{1/q}$.

We adapt the analysis in~\cite{shen2020online} to our PDLA framework for general online convex covering problems, and show that it is also consistent and robust for online covering problems with $\ell_q$-norm objectives, similarly in lieu of the monotone gradient assumptions in~\cite{azar2016online}. The performance of our algorithm for online covering problems with $\ell_q$-norm objectives is (informally) described by the following theorem:
\begin{theorem}[Informal]\label{thm:lqinformal}
Given an instance of the online covering problem with $\ell_q$-norm objectives \eqref{eq:lqcovering}, there exists an online algorithm that takes an advice $x' \in \nnreals^n$ and a confidence parameter $\lambda$. The algorithm is $O(\frac{1}{1-\lambda})$-consistent and $O(\log \frac{\kappa d}{\lambda})$-robust. Here, $\kappa$ is the condition number of the constraint matrix, and $d$ is the row sparsity of the constraint matrix.
\end{theorem}
The formal version of \Cref{thm:lqinformal} is \Cref{thm:lqmain} in \Cref{sec:Lq} which addresses the case when the advice $x'$ is infeasible.

\subsubsection{When is switching optimal?}
Designing simple solutions to natural problems of wide general interest is an ultimate goal of both theory (as they can be taught even in undergraduate courses!) and practice (as they can be implemented with a few lines of code, and would have strong provable guarantees!). However, while the study of learning-augmented algorithm is currently flourishing, many solutions and formulations are often somewhat ad-hoc, and hence there is a need for general frameworks, techniques, and paradigms. 

Inspired by the switching strategy noted in~\cite{grigorescu2022arxiv} and our own observations in this paper, we raise a basic question:
\begin{center}
    {\it Can one characterize the space of online problems augmented with advice, for which one may use classical online algorithms as a {\em black-box} to obtain {\em optimal} solutions?}
\end{center}
In particular, what are the necessary and sufficient conditions for such simple switching strategies to be close to being optimal? What features should a problem exhibit that would allow it to be solvable in a black-box fashion in the advice setting? Recently, \cite{daneshvaramoli2024competitive} showed that a similar switching-based approach can achieve optimal constants in online knapsack with succinct predictions, which also raises similar questions. We believe that a better understanding of this conceptual direction may lead to unifying frameworks in the study of algorithms with advice.

In this work, we make progress on sufficient conditions for the question above. In particular, we study online concave packing problems with advice and show that the switching framework is close to being optimal in this context. As a step towards necessary conditions, we show that switching strategies, at least ones found in~\cite{grigorescu2022arxiv} and our algorithmic framework for online concave packing problems, are not applicable to online convex covering problems, and instead requires more sophisticated methods to solve in general.

\subsubsection{Applications}

We explicitly study the application of our frameworks to well-motivated problems in both online concave packing and online convex covering settings in \Cref{sec:Appl}. Our algorithms match the state-of-the-art algorithms by setting $\lambda=1$ and can potentially outperform the best-known online algorithms when the advice is accurate.

We apply \Cref{thm:packinginformal} to a variety of packing problems, including knapsack, resource management benefit, throughput maximization, network utility maximization, and optimization with inventory constraints. 
We then present an application of Theorem \ref{thm:convexinformal} to online mixed covering and packing, and by extension a variety of sub-problems such as capacity-constrained facility location and capacitated multicast, and an application of Theorem \ref{thm:lqinformal} to online buy-at-bulk network design.

\subsection{Additional Prior Works} \label{sec:Prior}

\paragraph{Learning-augmented algorithms.} Learning-augmented algorithms have been extensively studied for many fundamental online problems. \cite{Rohatgi20} and~\cite{lykouris2021competitive} showed that accurate predictions can lead to competitive ratios better than classical impossibility results for online caching, and subsequent works studied problems such as set cover~\cite{bamas2020primal}, ski rental~\cite{purohit2018improving,bamas2020primal}, clustering~\cite{ErgunFSWZ22}, graph problems~\cite{banerjee2023graph,brand2024dynamic,henzinger2024complexity}, facility location~\cite{jiang2021online}, knapsack~\cite{IKMP2021,boyar2022online,daneshvaramoli2023online}, matching~\cite{DinitzILMV21,jin2022online}, and others~\cite{hsu2019learning,antoniadis2020secretary,BamasMRS20,bamas2020primal,lattanzi2020online,DiakonikolasKTV21,dutting2021secretaries,ChenEILNRSWWZ22,chen2022,cohen2023general,enikHo2023primal,brand2024dynamic}. While most works in learning-augmented algorithms treat the advice as a black-box device, there are some recent works that studies how to properly obtain such advices (e.g., \cite{DBLP:journals/corr/abs-2011-07177, khodak2022learning}). A comprehensive online archive of recent works in the field of learning-augmented algorithms can be found at~\cite{ALPS}.

Other lines of work exploring alternative forms of additional information given to the algorithm studies the stochastic setting, where the assumption is that the input instance is drawn from an underlying distribution known to the algorithm. Problems studied under this model include stochastic matching~\cite{Feldman18}, graph optimization~\cite{azar2022online}, and others~\cite{mahdian2012online,Mitzenmacher18}. In comparison, the advice model studied by the learning-augmented algorithm community at large considers only explicit advices about the future inputs of the instance. While these advices may take forms that admits a distributional interpretation, there is no assumption made on the input distribution itself.

\paragraph{The primal-dual method.} The primal-dual method, introduced in the seminal work of~\cite{GoemansW95}, is a powerful algorithmic technique used to solve a variety of problems in the field of approximation algorithms. The primal-dual method has been applied to individual problems such as online set cover~\cite{alon2009online}, network optimization~\cite{alon2006general}, ski rental~\cite{kmmo}, and generalized into a unifying framework for online LP-based problems by~\cite{buchbinder2009online}. We refer the readers to~\cite{buchbinder2009survey} for a survey on the topic of primal-dual methods for online algorithms.

In the field of learning-augmented algorithms, Bamas, Maggiori, and Svensson~\cite{bamas2020primal} initiated the study of primal-dual learning-augmented algorithms, and inspired follow-up work to extend the framework to more generalized models. \cite{anand2022online} considered PDLA algorithms on online covering LPs with multiple predictions; \cite{grigorescu2022learning} generalized the PDLA framework to general online covering problems with fractional constraints and advices, as well as online semidefinite programs~\cite{elad2016online}. 

Towards generalizing the framework of~\cite{buchbinder2009online} and~\cite{grigorescu2022learning} to non-linear objectives, a concurrent line of work by~\cite{thang2021online} and~\cite{kevi2023primaldual} studies online packing and covering problems using configuration programs and multilinear extensions of monotone objective functions. We remark that their model admits more generality in allowing non-convex or non-concave objective functions, but also incurs a potential loss in performance due to the generalized setting, and thus are incomparable to ours. Their advice is also integral, while the form of advice we consider allows for fractional predictions. 

Some prior works have studied problems that allow a convex covering (e.g., \cite{BamasMRS20,cohen2023general,golrezaei2023online,balkanski2024energy,lechowicz2024online}) and concave packing (e.g.,\cite{IKMP2021,dutting2021secretaries,jin2022online,enikHo2023primal}) program formulation in the context of learning-augmented algorithms. 
Our learning-augmented frameworks unifies and generalizes these settings and also applies to many other problems.

%% file: packing.tex
\section{Switching Algorithms for Online Concave Packing}\label{sec:Switch}

We present a simple, possibly folklore, learning-augmented online algorithm for concave packing, formulated as follows:
\begin{align*}
  \begin{aligned}
    \text{max } & g(y) 
    \text{ over } y \in \nnreals^m 
    \text{ subject to } A^T y \leq b.
  \end{aligned}
\end{align*}

Let $\cO$ be an $\alpha$-competitive $\beta$-feasible online algorithm for the packing problem \eqref{eq:concavepacking}, for some $\alpha, \beta \ge 1$.
The algorithm keeps track of two candidate solutions: the advice $y'$, and $y^\cO$ from the non-learning-augmented online algorithm $\cO$, and combines them. Whenever a row of $A$ arrives online, the algorithm obtains a solution from both the online algorithm and the advice and sets $y_i$ to a weighted interpolation between the two values, $\lambda \cdot y_i^{\cO} + (1-\lambda) \cdot y_i'$. Once the algorithm finds out that the advice violates any constraint by a factor of $\beta$, it discards the advice for this round and follows the online algorithm only. The algorithm is presented in \Cref{alg:simple}.

\begin{algorithm}[!htb]
\caption{A Simple Learning-Augmented Online Algorithm for the Packing Problem \eqref{eq:concavepacking}} \label{alg:simple}
\renewcommand{\algorithmicrequire}{\textbf{Input:}}
\renewcommand{\algorithmicensure}{\textbf{Output:}}

\algorithmicrequire{ Any online algorithm algorithm $\cO$, advice $y'$, confidence parameter $\lambda$.}

\algorithmicensure{ The online solution $y$.}
\begin{algorithmic}[1]
\For{$i = 1, 2, ...$} \Comment{each arriving row or constraint}
    \State Update $A$ by adding a new row $i$.
    \State Run $\cO$ for round $i$ and obtain $y_i^{\cO}$.
    \If{$A^T y' \le \beta b$} \label{line:alg-simple-case-1}\Comment{The advice is approximately feasible} 
        \State $y_i \gets \lambda y_i^{\cO} + (1-\lambda) y_i'$. \label{line:xj1}
    \Else
    \label{line:alg-simple-case-2}
        \State $y_i \gets y_i^{\cO}$. \label{line:xj2}
    \EndIf
\EndFor
\end{algorithmic}
\end{algorithm}

We formally characterize the performance of \Cref{alg:simple} in \Cref{thm:packingmain}, presented below:

\begin{theorem}\label{thm:packingmain}
For the learning-augmented online concave packing problem, there exists an online algorithm that takes an $\alpha$-competitive $\beta$-feasible online algorithm $\cO$, an advice $y' \in \nnreals^n$, and a confidence parameter $\lambda \in [0, 1]$, and generates a solution $\bar y$ satisfying $A^T \bar y \leq (2 - \lambda) \beta b$ such that $g(\bar y) \geq \frac{\lambda}{\alpha} \opt$.
Additionally, if $y'$ is $\beta$-feasible, i.e., $A^T y' \leq \beta b$, then $g(\bar y) \geq (1-\lambda) g(y')$.
\end{theorem}

\begin{proof}
Denote by $y'_{trim}$ as the advice $y'$ with all entries discarded by \Cref{alg:simple} set to 0 instead. The final packing solution $\bar y$ is at least $\lambda y^\cO + (1 - \lambda) y'_{trim}$. By the monotonicity of $g$, we have also $g(\bar y) \geq g(\lambda y^\cO + (1 - \lambda) y'_{trim})$. It immediately holds that $g(\bar y) \geq g(\lambda y^\cO) \geq \lambda \cdot g(y^\cO) \geq \frac{\lambda}{\alpha} \opt$, since $g$ is concave and $g(\bfzero) = 0$.

If $y'$ is $\beta$-feasible, the algorithm will not trim any entry of $y'$, and it similarly holds that $g(\bar y) \geq (1-\lambda) g(y'_{trim}) = (1-\lambda) g(y')$.

For feasibility, observe that $A^T y'_{trim} \leq \beta b$ and $A^T y^\cO \leq \beta b$. Thus, $A^T \bar y \leq A^T (1 - \lambda) y'_{trim} + A^T y^\cO \leq (2 - \lambda) \beta b$, as desired.
\end{proof}

\begin{remark}
We note that depending on the problem of interest, $\beta$ might not be a fixed parameter. For online packing linear programs, $\beta$ is a conditional number non-decreasing over time, which depends on $A$, so it is not necessarily the case that the advice is only used in earlier rounds. See \Cref{sec:ApplPacking} for a more detailed discussion. 
\end{remark}

%% file: convex.tex
\section{Primal-Dual Framework for Online Convex Covering}\label{sec:Convex}

In this section we present a consistent and robust primal-dual learning-augmented (PDLA) algorithm for online covering problems with convex objectives, based on the primal-dual framework of~\cite{buchbinder2009online} and its extension to the learning-augmented setting in~\cite{grigorescu2022learning}. In the classical online setting, online convex covering problems are studied in~\cite{buchbinder2014online} and subsequently~\cite{azar2016online}, in which the authors presented an extension of~\cite{buchbinder2009online}'s primal-dual framework to convex non-linear objective functions; our PDLA framework take inspiration from their work as well.

Recall that our goal is to solve the following problem in an online fashion:
\begin{align*}
  \begin{aligned}
    \text{min } & f(x) 
    \text{ over } x \in \nnreals^n 
    \text{ subject to } A x \geq \bfone.
  \end{aligned} \label{eq:convexcovering}
\end{align*}
where $f : \nnreals^n \to \nnreals$ is a convex, monotone, differentiable function, and $A \in \nnreals^{m \times n}$ is the constraint matrix where each of the $m$ rows correspond to a covering constraint. We make additional assumptions about $f$, that $f(\bfzero) = 0$, $\nabla f$ is monotone, and there exists some $p := \sup_{x \geq \bfzero} \frac{\langle x, \nabla f(x) \rangle}{f(x)}$.

The primal-dual method is a class of algorithms that simultaneously consider the dual packing program. Specifically, we choose the Fenchel dual program formulated as:
\begin{equation}\label{eq:convexpacking}
    \text{max } \sum_{i=1}^m y_i - f^*(\mu) \text{ over } y \in \nnreals^m \text{ subject to } y^T A \leq \mu^T
\end{equation}
where $f^* : \nnreals^n \to \nnreals$ is the Fenchel dual of $f$, defined as $f^*(\mu) = \sup_z \{ \mu^T z - f(z) \}$. $f^*$ is always convex; for more properties of Fenchel dual functions, we refer the readers to~\cite{borwein2006convex}.

As with covering and packing linear programs, the convex covering program and its Fenchel dual also exhibits a form of duality:
\begin{restatable}[\cite{buchbinder2014online}]{fact}{Duality}\label{fact:duality}
Let $x$ and $(y, \mu)$ be any feasible solution to the primal covering program \eqref{eq:convexcovering} and the dual packing program \eqref{eq:convexpacking}, respectively. Then
\[f(x) \geq \sum_{i=1}^m y_i - f^*(\mu)\]
\end{restatable}

\begin{proof}
Observe that by definition,
\[\sum_{i=1}^m y_i = y^T \bfone \leq y^T A x \leq \mu^T x = (\mu^T x - f(x)) + f(x) \leq f^*(\mu) + f(x).\]
Rearranging, we have $\sum_{i=1}^m y_i - f^*(\mu) \leq f(x)$, as desired.
\end{proof}

Our algorithm maintains a variable $\tau$ denoting the continuous flow of time, and upon the arrival of each constraint $t$, increment each covering variable $x_j$ for $j \in [n]$ with differing rates, dependent on the advice $x'$, until the constraint is satisfied. We increment the dual packing variable $y_t$ as well, but also potentially decrement some other packing variable to ensure that the packing solution is feasible.

To paint a picture of our strategy in more detail, let $A_i$ denote the $i$-th row of $A$. In round $t$, upon the arrival of constraint $t$ and row $A_t$, we check if the advice is feasible for this constraint. If the advice is feasible, i.e., $A_t x' \geq 1$, we increment each variable $x_j$ with rate
\[\frac{a_{tj}}{\nabla_j f(x^{(\tau)})} \left(x_j^{(\tau)} + \frac{\lambda}{a_{tj} d} + \frac{(1-\lambda) x'_j \bfone_{x_j^{(\tau)} < x'_j}}{A_t x'_c}  \right)\]
where $x_j^{(\tau)}$ and $x^{(\tau)}$ denote the value of $x_j$ and $x$ at time $\tau$, respectively, and $d$ is the row sparsity of $A$.
$x'_c$ is the advice restricted to entries in which the corresponding variables has not reached the advice yet. Equivalently, the $j$-th entry of $x'_c$ is equal to the $j$-th entry of $x'$ if $x_j^{(\tau)} < x'_j$, and $0$ otherwise. $\bfone_{x_j^{(\tau)} < x'_j}$ is the indicator variable with value $1$ if $x_j^{(\tau)} < x'_j$ and $0$ otherwise.

Intuitively, the two additive terms, $\frac{\lambda}{a_{tj} d}$ and $\frac{(1-\lambda) x'_j \bfone_{x_j^{(\tau)} < x'_j}}{A_t x'_c}$, corresponds to the contribution of the classical online algorithm and the advice to the growth rate, respectively. 
Our algorithm applies the additional growth rate attributed to the advice if the corresponding variable $x_j$ has not reached the value suggested by the advice $x'_j$ yet. If $x_j$ has reached the suggested value, our algorithm instead zeroes the additional growth rate, and relies only on the component attributed to the online primal-dual method.
We remark that in~\cite{grigorescu2022learning}, the contribution of the classical online primal-dual method corresponds to the term $\frac{\lambda}{A_i \bfone}$ instead. While this term is normalized and invariant across all $j \in [n]$, it implicitly introduced a dependence on the values of all primal variables to the growth rate of $x_j$, even for variables $x_{j'}$ with $j' \neq j$. Our modification removes this dependence, allowing us to tighten the analysis of our algorithm and obtain better robustness bounds over~\cite{grigorescu2022learning}.

If the advice is not feasible, i.e., $A_t x' < 1$, we disregard the advice temporarily, and increment $x_j$ with rate
\[\frac{a_{tj}}{\nabla_j f(x^{(\tau)})} \left(x_j^{(\tau)} + \frac{1}{a_{tj} d} \right)\]
In this case, our rate of growth coincides with that of both the classical, non-learning-augmented online primal-dual method, and the case when the advice is feasible above with $\lambda$ set to 1. Since the advice is infeasible and cannot provide valuable insight, we consider only the contribution from the online algorithm for constraint $t$.

The increment scheme for the dual variables is as follows. In round $t$, we initialize the packing variable $y_t$ corresponding to the arriving constraint $t$ to $0$ initially, and increment it with rate
\[r := \frac{\partial y_t^{(\tau)}}{\partial \tau} = \frac{\delta}{\log (1 + \frac{2}{\lambda} d^2)}\]
for some value of $\delta$ we choose later that depends only on $p, d$, and $\lambda$.

Combining the ideas described above suffices to yield a $O(\log \frac{\kappa d}{\lambda})$-robust learning-augmented algorithm, where $\kappa$ is the condition number of the constraint matrix $A$, defined as the ratio between the largest and the smallest non-zero entry in $A$. However, since the entries of $A$ arrives online, $\kappa$ is unknown to the algorithm up front and can be arbitrarily large. We follow the idea of \cite{buchbinder2014online,azar2016online} to additionally decrease some of the packing variables, to ensure that each dual packing constraint in \eqref{eq:convexpacking} is exactly satisfied. Specifically, for each dual packing constraint $j$ (corresponding to the covering variable $x_j$) that is tight, i.e., $\sum_{i=1}^t a_{ij} y_i = \mu_j$, we identify the variable $y_{m^*_j}$ with the largest coefficient $a_{(m^*_j)j}$, and decrement it with rate
\[ \frac{\partial y_{m^*_j}}{\partial \tau} = - \frac{a_{tj}}{a_{(m^*_j)j}} \cdot r \]
This ensures that the net growth rate of $\sum_{i=1}^t a_{ij} y_i$ is $0$, so that the constraint is not violated.

Note that the technique of decreasing the dual variables $y$ is unusual, in that it violates the monotonicity of the dual variables, which is a common feature of primal-dual methods, such as in \cite{buchbinder2007online, bamas2020primal, grigorescu2022learning}. The purpose of this property is for the primal-dual method to be a general-purpose method that solves both covering programs and packing programs simultaneously. However, in our context, we are exclusively concerned with the online convex covering problem, and not its dual packing problem. In fact, learning-augmented algorithms for covering and packing problems that are usually dual to each other often have differing specification and advice forms, and do not share simple and elegant algorithms that serves both problems.

For the dual variables $\mu$, we set $\mu$ to $\delta \nabla f(x)$, which is monotone due to the monotonicity of $x$ and the assumption that $\nabla f(x)$ is monotone as well.

We present our learning-augmented primal-dual algorithm for online convex covering problems during round $t$ in \Cref{alg:ConvexCover}. While the algorithm description increments $x$ and $y$, and sets $\mu$ simultaneously during execution, the update rule for the primal covering variables, $x$, in fact does not depend on $y$ or $\mu$. Since we are not concerned about maintaining the dual solution, $(y, \mu)$, and uses it for purely analytical purposes, we can instead increment $y$ and set $\mu = \delta \nabla f(\bar x)$ in hindsight after computing the final solution $\bar x$ to the primal covering problem~\eqref{eq:convexcovering}, simplifying the algorithm and its analysis.

\begin{algorithm}[!htb]
\caption{PDLA algorithm for online convex covering problems in round $t$} \label{alg:ConvexCover}

\renewcommand{\algorithmicrequire}{\textbf{Input:}}
\renewcommand{\algorithmicensure}{\textbf{Output:}}

\algorithmicrequire{ $x_1, \cdots, x_n$, $y_1, \cdots, y_{t-1}$, $\mu$: current solution, $A_{tj}$ for $1 \leq j \leq n$: current coefficients, $(x', \lambda)$: advice and confidence parameter.}

\algorithmicensure{ Updated $x$ and $(y, \mu)$.}

\begin{algorithmic}[1]

    \State $\mu \gets \delta \nabla f(\bar x)$. \Comment{$\delta$ is a parameter we determine and define later}
    \State $y_t \gets 0$.
    \State $\tau \gets \text{ current time}$.
    \If {$A_t x' \geq 1$} \Comment{if $x'$ is feasible}
        \State $D_j^{(t)} \gets \frac{\lambda}{a_{tj} d} + \frac{(1-\lambda) x'_j \bfone_{x_j < x'_j}}{A_i x'_c}$.
    \Else
        \State $D_j^{(t)} \gets \frac{1}{a_{tj} d}$.
    \EndIf
    \While {$A_t x < 1$} \Comment{while unsatisfied}
        \State Increase $\tau$ at rate $1$. 
        \For {each $j \in [n]$ s.t. $a_{tj} > 0$}
            \State Increase $x_j$ at rate
            \[\frac{\partial x_j}{\partial \tau} = \frac{a_{tj}}{\nabla_j f(x)} (x_j + D_j^{(t)})\]
        \EndFor
        \State Increase $y_t$ at rate
        \[r := \frac{\partial y_t}{\partial \tau} = \frac{\delta}{\log (1 + \frac{2}{\lambda} d^2)}\]
        \For {each $j \in [n]$ s.t. $\sum_{i=1}^t a_{ij} y_i = \mu_j$} \Comment{decrement some variables to maintain dual feasibility}
            \State $m^*_j = \arg \max_{i=1}^t \{a_{ij} | y_i > 0\}$.
            \State Decrease $y_{m^*_j}$ at rate $\frac{a_{tj}}{a_{(m^*_j)j}} \cdot r$.
        \EndFor
    \EndWhile
\end{algorithmic}
\end{algorithm}
While the algorithm increments the variables in a continuous fashion, it is possible to be implemented discretely, up to any desired precision, via binary searches.

We state the performance of \Cref{alg:ConvexCover} in \Cref{thm:convexmain}, presented below:

\begin{restatable}{theorem}{ConvexMain}\label{thm:convexmain}
For the learning-augmented online convex covering problem with the monotone gradient assumption, there exists an online algorithm that takes a problem instance and an advice $x'$, and generates a solution $\bar x$ such that 
\[f(\bar x) \leq \min \left\{ O \left( \frac{1}{1-\lambda} \right) f(x') + O((p \log d)^p) \opt, O\left( \left( p \log \frac{d}{\lambda} \right)^p \right) \opt \right\} \]
Additionally, if $x'$ is feasible, i.e., $A x' \geq \bfone$, then $f(\bar x) \leq O\left(\frac{1}{1-\lambda}\right) f(x')$. Here, $d = \max_{i \in [m]} |\{a_{ij} | a_{ij} > 0\}|$ is the row sparsity of the constraint matrix $A$, and $p := \sup_{x \geq \bfzero} \frac{\langle x, \nabla f(x) \rangle}{f(x)}$.
\end{restatable}

We prove \Cref{thm:convexmain} by proving the two components, $f(\bar x) \leq O((p \log d/\lambda)^p) \opt$ (robustness) and $f(\bar x) \leq O(1/(1-\lambda)) f(x') + O((p \log d)^p)$ (consistency), separately.

\paragraph{Robustness.} First, we prove the robustness of \Cref{alg:ConvexCover}, that $f(\bar x) \leq O((p \log d/\lambda)^p) \opt$. Specifically, we show the following:
\begin{enumerate}
    \item $\bar x$ is feasible and monotone;
    \item $(\bar y, \mu)$ is feasible, and $\mu$ is monotone;
    \item The primal objective $P$ is at most $O((p \log d/\lambda)^p)$ times the dual objective $D$.
\end{enumerate}
Equipped with these subclaims and weak duality, we have
\[P = f(\bar x) \leq O((p \log d/\lambda)^p) D \leq O((p \log d/\lambda)^p) \opt \]
as desired.

We begin by showing the first two feasibility claims together.
\begin{lemma}\label{lem:convexfeasibility}
    For any $\delta > 0$, the following are maintained.
    \begin{itemize}
        \item The algorithm maintains a feasible monotone primal solution.
        \item The algorithm maintains a feasible dual solution with monotone $\mu_j$.
    \end{itemize}
\end{lemma}

\begin{proof}
The feasibility of the primal solution follows by construction, since we only increment each coordinate of $x$ continuously, at least one of which contributing with non-zero rate, and only stop when the arriving constraint is satisfied.

By design we only set $\mu$ once throughout the entire execution. To show that the dual is feasible, observe that at round $t$ and time $\tau$, for any dual constraint $\sum_{i=1}^t a_{ij} y_i^{(\tau)} \leq \mu_j$, if the inequality is strict, the constraint is satisfied and we are done. If there is equality, line 11 through 13 of \Cref{alg:ConvexCover} will counteract by decrementing $y_{m^*_j}$, resulting in a net growth rate of
\[\frac{\partial}{\partial \tau} \sum_{i=1}^t a_{ij} y_i^{(\tau)} = a_{ij} \cdot r - a_{(m_j^*) j} \cdot  \frac{a_{tj}}{a_{(m^*_j)j}} \cdot r = 0\]
so that the constraint remains satisfied.
\end{proof}

Towards showing that the ratio between the primal and the dual objectives are bounded, our strategy on a high level is classical to proving the competitiveness of primal-dual algorithms, bounding the ratio between the rates of change of both the primal and the dual objectives during the algorithm's execution. However, the rate of change of the dual objective, specifically $\sum_i y_i$, is very nuanced to analyze, since we both increment $y_t$ on line 11 of \Cref{alg:ConvexCover} and decrement some other dual variable on line 14. Towards this, we use a sandwich strategy with the primal variables $x_j$ as a proxy, to bound the rate of decrease of $\sum_i y_i$ on line 14.

We present the following lemma as the first part of our sandwich strategy, a lower bound on the value of $x_j$ at time $\tau$:
\begin{lemma}\label{lem:convexxjbound}
    For a variable $x_j$, let $T_j = \{i | a_{ij} > 0\}$ and let $S_j$ be any subset of $T_j$. Then for any $t \in T_j$ and $\tau_t \leq \tau \leq \tau_{t+1}$,
    \begin{align*}
        x_j^{(\tau)} \geq\; & \frac{\lambda}{\max_{i \in S_j} \{a_{ij}\} \cdot d}  \cdot \left( \exp \left( \frac{\log (1 + \frac{2}{\lambda}d^2)}{\mu_j} \sum_{i \in S_j} a_{ij} y_i^{(\tau)} \right) - 1 \right)
    \end{align*}
    where $\tau_t$ denotes the value of $\tau$ at the arrival of the $t$-th primal constraint.
\end{lemma}
\begin{proof}
Note that 
\begin{align}
\frac{\partial x_j}{\partial y_t} &= \frac{\partial x_j}{\partial \tau} \cdot \frac{\partial \tau}{\partial y_t} = \frac{a_{tj} (x_j + D_j^{(t)})}{\nabla_j f(x)} \cdot \frac{\log (1 + \frac{2}{\lambda}d^2)}{\delta} \nonumber\\
&\geq \log (1 + \frac{2}{\lambda}d^2) \frac{a_{tj} (x_j + D_j^{(t)})}{\delta \nabla_j f(\bar x)} \label{eq:dxdy-mono-grad}
\end{align}
where the last inequality is due to our assumption that the gradient is monotone, and that $x \leq \bar x$.
Solving this differential equation, we have
\[\frac{x_j^{(\tau)} + D_j^{(t)}}{x_j^{(\tau_t)} + D_j^{(t)}} \geq \exp \left( \frac{\log (1 + \frac{2}{\lambda}d^2)}{\delta \nabla_j f(\bar x)} \cdot a_{tj} y_t^{(\tau)} \right)\]

From here, we use the following two claims:
\begin{claim}\label{claim:telescope}
For all scalars $a \geq b > 0$ and $c_1 \geq c_2 \geq 0$,
\[\frac{a+c_1}{b+c_1} \leq \frac{a+c_2}{b+c_2}\]
\end{claim}
\begin{claim}\label{claim:Dj}
For any $j \in [n]$ and $i \in S_j$ the following holds:
\[D_j^{(i)} = \frac{\lambda}{a_{ij} d} + \frac{(1-\lambda)x'_j \bfone_{x_j < x_j'}}{A_i x'_c} \geq \frac{\lambda}{\max_{i \in S_j} a_{ij} d}\]
\end{claim}

Multiplying over all indices, where for notational convenience we set $\tau_{t+1} = \tau$,
\begin{align*}
\exp \left(\frac{\log (1 + \frac{2}{\lambda}d^2)}{\mu_j} \sum_{i\in S_j} a_{ij} y_i^{(\tau)} \right)
&= \exp \left( \sum_{i \in S_j} \frac{\log (1 + \frac{2}{\lambda}d^2)}{\delta \nabla_j f(\bar x)} \cdot a_{ij} y_i^{(\tau)} \right)\\
&\leq \prod_{i \in S_j} \frac{x_j^{(\tau_{i+1})} + D_j^{(i)}}{x_j^{(\tau_i)} + D_j^{(i)}} \\
&\leq \prod_{i \in S_j} \frac{x_j^{(\tau_{i+1})} + \frac{\lambda}{\max_{i \in S_j} a_{ij} d}}{x_j^{(\tau_i)} + \frac{\lambda}{\max_{i \in S_j} a_{ij} d}} \qquad \text{ by \Cref{claim:telescope} and \Cref{claim:Dj}}\\
&\leq \prod_{i \in T_j} \frac{x_j^{(\tau_{i+1})} + \frac{\lambda}{\max_{i \in S_j} a_{ij} d}}{x_j^{(\tau_i)} + \frac{\lambda}{\max_{i \in S_j} a_{ij} d}} \qquad \text{ since $x_j$ is monotone}\\
&= \frac{x_j^{(\tau)} + \frac{\lambda}{\max_{i \in S_j} a_{ij} d}}{\frac{\lambda}{\max_{i \in S_j} a_{ij} d}} \qquad \text{ by a telescoping argument over all indices}
\end{align*}
Reorganizing the terms yields the desired bound on $x_j^{(\tau)}$.
\end{proof}

Since the arriving constraint is not satisfied yet during the algorithm, trivially $1/a_{tj}$ is a upper bound on $x_j$. Equipped with this and the lower bound in \Cref{lem:convexxjbound}, we move on to bound the ratio between the primal objective $P$ and the dual objective $D$:
\begin{lemma}\label{lem:convexrobust}
Let $p := \sup_{x \geq \mathbf{0}} \frac{\langle x, \nabla f(x) \rangle}{f(x)}$. Then
\begin{align*}
    P &= f(\bar x) \leq (4 p \log (1 + \frac{2}{\lambda}d^2))^p \left(\sum_{i=1}^m y_i - f^*(\mu) \right) = O((p \log \frac{d}{\lambda})^p) D
\end{align*}
\end{lemma}

\begin{proof}
We assume that $x'$ is feasible, i.e., $A_t x' \geq 1$. A similar analysis can be obtained for the case of $A_t x' < 1$ by setting $\lambda = 1$.

Consider the update when primal constraint $t$ arrives and $\tau$ is the current time. Let $U(\tau)$ denote the set of tight dual constraints at time $\tau$, i.e., for every $j \in U(\tau)$ we have $a_{tj} > 0$ and $\sum_{i=1}^t a_{ij} y_i^{(\tau)} = \delta \nabla_j f(\bar x)$. $|U(\tau)| \leq d$, and define for every $j$ the set $S_j := \{i | a_{ij} > 0, y_i^{(\tau)} > 0\}$. 
Clearly $\sum_{i \in S_j} a_{ij} y_i^{(\tau)} = \sum_{i=1}^t a_{ij} y_i^{(\tau)} = \delta \nabla_j f(\bar x)$, and $\sum_j a_{tj} x_j^{(\tau)} < 1$, thus by \Cref{lem:convexxjbound}, we have
\begin{align*}
    \frac{1}{a_{tj}} > x_j^{(\tau)} \geq\;& \frac{\lambda}{\max_{i \in S_j} \{a_{ij}\} \cdot d} \cdot \left( \exp \left( \log (1 + \frac{2}{\lambda}d^2) \right) - 1 \right)
\end{align*}
Rearranging, we have
\[\frac{a_{tj}}{a_{m^*_j j}} = \frac{a_{tj}}{\max_{i \in S_j} a_{ij}} \leq \frac{1}{2d}\]
Thus, we can bound the rate of change of $\sum_{i=1}^t y_i$ at time $\tau$:
\begin{align*}
    \frac{\partial (\sum_{i=1}^t y_i)}{\partial \tau} &\geq r - \sum_{j \in U(\tau)} \frac{a_{tj}}{a_{m^*_j j}} \cdot r \geq r \left( 1 - \sum_{j \in U(\tau)} \frac{1}{2d} \right) \geq \frac{1}{2} r
\end{align*}
where the last inequality follows from $|U(\tau)| \leq d$.

On the other hand, when processing constraint $t$, the rate of increase of the primal is
\begin{align*}
\frac{\partial f(x^{(\tau)})}{\partial \tau} &= \sum_j \nabla_j f(x^{(\tau)}) \frac{\partial x_j^{(\tau)}}{\partial \tau} \qquad \text{ by the chain rule}\\
&= \sum_{j | a_{tj} > 0} \nabla_j f(x^{(\tau)}) \left( \frac{a_{tj} (x_j^{(\tau)} + D_j^{(t)})}{\nabla_j f(x^{(\tau)})} \right)\\
&= \sum_{j | a_{tj} > 0} \left(a_{tj} x_j^{(\tau)} + \frac{\lambda}{d} + \frac{(1-\lambda) a_{tj} x'_j \bfone_{x_j < x'_j}}{A_t x'_c}\right)\\ 
&\leq 1 + \lambda + (1 - \lambda) = 2 
\end{align*}
Here, the last inequality is due to the constraint being unsatisfied, and the definition of $x'_c$ as the constraint restricted to entries where $x_j < x'_j$ holds.

Thus, bounding the rate of change between the primal and the dual:
\[\frac{\partial (\sum_{i=1}^t y_i^{(\tau)})}{\partial f(x^{(\tau)})} \geq \frac{1}{4} r = \frac{\delta}{4 \log (1 + \frac{2}{\lambda} d^2)} \]

Let $\bar{x}$ and $\bar{y}$ be the final primal and dual solutions. We have
\begin{equation*} 
    \sum_{i=1}^m \bar y_i \geq \frac{\delta}{4 \log (1 + \frac{2}{\lambda}d^2)} \cdot f(\bar x).
\end{equation*}

We have thus obtained our bound on the part of the dual objective corresponding to the variables $y$. To obtain a simple form of the dual objective corresponding to $\mu$, we use the following facts on $f$, which bounds how quickly $f$ and $f^*$ can grow:
\begin{fact}[Bounded Growth, \cite{azar2016online}]\label{fact:boundedgrowth}
    Let $f$ be a convex, monotone, differentiable function with non-decreasing gradient and $f(0) = 0$. Assume additionally that there is some $p \geq 1$ s.t. $x \nabla f(x) \leq p f(x)$ for all $x \in \nnreals^n$. Then,
    \begin{enumerate}
        \item For all $\delta \geq 1$, for all $x \in \nnreals^n$, $f(\delta x) \leq \delta^p f(x)$;
        \item For all $x \in \nnreals^n$, $f^*(\nabla f(x)) = x^T \nabla f(x) - f(x) \leq (p-1) f(x)$;
        \item If $p > 1$ then for any $0 < \gamma \leq 1$, $x \in \nnreals^n$, $f^*(\gamma x) \leq \gamma^{\frac{p}{p-1}} f^*(x)$ and $f^*(\gamma \nabla f(x)) \leq \gamma^{\frac{p}{p-1}} (p-1) f(x)$;
        \item If $p = 1$ then for any $0 < \gamma \leq 1$, $x \in \nnreals^n$, $f^*(\gamma \nabla f(x)) \leq 0$.
    \end{enumerate}
\end{fact}
See \cite{azar2016online} for a proof of \Cref{fact:boundedgrowth}.

By definition of the dual objective $D$, we have
\begin{align*}
    D &= \sum_{i=1}^m \bar y_i - f^*(\mu)\\
    &\geq \frac{\delta}{4 \log (1 + \frac{2}{\lambda}d^2)} \cdot f(\bar x) - f^*(\delta \nabla f(\bar x))\\
    &\geq \left( \frac{\delta}{4\log (1 + \frac{2}{\lambda}d^2)} - \delta^{\frac{p}{p-1}} (p-1)\right) \cdot f(\bar x)
\end{align*}
where the last inequality is from (3) of \Cref{fact:boundedgrowth}. 
With the choice of $\delta = \frac{1}{(4 p \log (1 + \frac{2}{\lambda}d^2))^{p-1}}$, we have
\begin{align*}
\frac{\delta}{4\log (1 + \frac{2}{\lambda}d^2)} - \delta^{\frac{p}{p-1}} (p-1) &= \frac{1}{(4 \log (1 + \frac{2}{\lambda} d^2))^p \cdot p^{p-1}} - \frac{p-1}{(4 \log (1 + \frac{2}{\lambda} d^2))^p \cdot p^p} \\
&= \frac{1}{(4p \log (1 + \frac{2}{\lambda} d^2))^p}
\end{align*}

Thus,
\[P = f(\bar x) \leq (4 p \log (1 + \frac{2}{\lambda}d^2))^p D = O((p \log \frac{d}{\lambda})^p) D\]
as desired.
\end{proof}

\paragraph{Consistency.} Towards showing the consistency of \Cref{alg:ConvexCover}, we employ a similar proof strategy as in \cite{bamas2020primal} and~\cite{grigorescu2022learning}. We partition the growth rate of the primal variables $x_j$ based on whether the variable has exceeded the corresponding entry in the advice $x'_j$. With this partition, we argue that the part of the growth rate credited to the classical primal-dual increment procedure is at most a certain factor of the part of the growth rate credited to the advice. Specifically, we present and prove the following lemma:

\begin{lemma}\label{lem:convexconsistency}
\Cref{alg:ConvexCover} is $O(\frac{1}{1-\lambda})$-consistent, i.e., if $\bar{x}$ is the final solution of the algorithm, then $f(\bar{x}) \leq O(\frac{1}{1-\lambda}) f(x') + O((p \log d)^p) OPT$.
\end{lemma}
\begin{proof}
Let $\tau$ denote the current time and $t$ denote the current round.
For any arriving constraint $A_t$, in the case of $A_t x' < 1$, we can bound the growth rate similarly to in \Cref{lem:convexrobust} with $\lambda$ set to $1$. Thus we can obtain $f(\bar x) \leq O((p \log d)^p) OPT$.

Assuming that $x'$ is feasible, i.e., $A_t x' \geq 1$, notice that
\[\frac{\partial f(x^{(\tau)})}{\partial \tau} = \sum_j \nabla_j f(x^{(\tau)}) \frac{\partial x_j^{(\tau)}}{\partial \tau}\]
We split the rate of growth of the primal objective as follows: let $S_c := \{j | x_j < x'_j\}$ denote the set of indices for which the primal variable has not reached the corresponding advice variable yet, and let $S_u := \{j | x_j \geq x'_j\}$ denote the set of indices for which the primal variable has surpassed the corresponding advice variable. Further define
\[r_c = \sum_{j \in S_c} \nabla_j f(x^{(\tau)}) \frac{\partial x_j^{(\tau)}}{\partial \tau} \]
and
\[r_u = \sum_{j \in S_u} \nabla_j f(x^{(\tau)}) \frac{\partial x_j^{(\tau)}}{\partial \tau} \]
with 
\[\frac{\partial f(x^{(\tau)})}{\partial \tau} = r_c + r_u\]
Note that the value of the advice, $f(x')$, can be lower bounded by the value of the solution of the primal-dual algorithm \emph{if it does not increase any variable past the advice}. This is due to the monotonicity of $f$. As a result, we can exclusively attribute the rate of increase in $r_c$ to the corresponding part in $f(x')$.

Notice that
\begin{align*}
    r_c &= \sum_{j \in S_c} \nabla_j f(x^{(\tau)}) \frac{\partial x_j^{(\tau)}}{\partial \tau}\\
    &= \sum_{j | x_j < x'_j} \left( a_{tj} x_j^{(\tau)} + \frac{\lambda}{d} + \frac{(1 - \lambda) a_{tj} x'_j \bfone_{x_j < x'_j}}{A_t x'_c} \right)\\
    &\geq 0 + 0 + (1 - \lambda)
\end{align*}
since $x'_c$ is defined as the advice restricted to indices in $S_c$, and
\begin{align*}
    r_u &= \sum_{j \in S_u} \nabla_j f(x^{(\tau)}) \frac{\partial x_j^{(\tau)}}{\partial \tau}\\
    &= \sum_{j | x_j \geq x'_j} \left( a_{tj} x_j^{(\tau)} + \frac{\lambda}{d} + \frac{(1 - \lambda) a_{tj} x'_j \bfone_{x_j < x'_j}}{A_t x'_c} \right)\\
    &\leq 1 + \lambda + 0
\end{align*}
since $d$ is the row sparsity, and that the arriving constraint $t$ has not been satisfied yet.
Thus we have 
\[r_u \leq \frac{1+ \lambda}{1-\lambda} r_c\]
and that
\[\frac{\partial f(x^{(\tau)})}{\partial \tau} \leq \left( 1 + \frac{1+\lambda}{1-\lambda} \right) r_c = O\left(\frac{1}{1-\lambda}\right) r_c\]
which, taking integrals over $\tau$, implies that $f(\bar x) \leq O\left(\frac{1}{1-\lambda}\right) f(x')$, as desired.
\end{proof}

Combining the robustness and consistency of \Cref{alg:ConvexCover}, we obtain a proof of \Cref{thm:convexmain}.

\subsection{A variant of \Cref{alg:ConvexCover} in~\cite{buchbinder2014online}}

We remark that our \Cref{alg:ConvexCover} follows the algorithmic structure and analysis of~\cite{azar2016online}, which is a merge of~\cite{buchbinder2014online} with other works. \cite{azar2016online} is a significant simplification of~\cite{buchbinder2014online}, suffering a potential penalty in the competitive ratio in exchange for much cleaner expressions and analyses. For completeness, we state our variant of \Cref{alg:ConvexCover} and \Cref{thm:convexmain} following~\cite{buchbinder2014online}, and briefly sketch the proof, specifically necessary modifications from the proof of \Cref{thm:convexmain}. In \Cref{sec:Appl}, we use this variant to obtain tighter bounds on various applications of our framework.

\begin{algorithm}[!htb]
\caption{The~\cite{buchbinder2014online} variant of \Cref{alg:ConvexCover} in round $t$} \label{alg:ConvexCoverComplex}

\renewcommand{\algorithmicrequire}{\textbf{Input:}}
\renewcommand{\algorithmicensure}{\textbf{Output:}}

\algorithmicrequire{ $x_1, \cdots, x_n$, $y_1, \cdots, y_{t-1}$, $\mu$: current solution, $A_{tj}$ for $1 \leq j \leq n$: current coefficients, $(x', \lambda)$: advice and confidence parameter.}

\algorithmicensure{ Updated $x$ and $(y, \mu)$.}

\begin{algorithmic}[1]

    \State $\mu \gets \nabla f(\delta \bar x)$.
    \State $y_t \gets 0$.
    \State $\tau \gets \text{ current time}$.
    \If {$A_t x' \geq 1$} \Comment{if $x'$ is feasible}
        \State $D_j^{(t)} \gets \frac{\lambda}{a_{tj} d} + \frac{(1-\lambda) x'_j \bfone_{x_j < x'_j}}{A_i x'_c}$.
    \Else
        \State $D_j^{(t)} \gets \frac{1}{a_{tj} d}$.
    \EndIf
    \While {$A_t x < 1$} \Comment{while unsatisfied}
        \State Increase $\tau$ at rate $1$. 
        \For {each $j \in [n]$ s.t. $a_{tj} > 0$}
            \State Increase $x_j$ at rate
            \[\frac{\partial x_j}{\partial \tau} = \frac{a_{tj}}{\nabla_j f(x)} (x_j + D_j^{(t)})\]
        \EndFor
        \State Increase $y_t$ at rate
        \[r := \frac{\partial y_t}{\partial \tau} = \frac{\min_{\ell=1}^n \left\{ \frac{\nabla_\ell f(\delta \bar x)}{\nabla_\ell f(\bar x)} \right\}}{\log (1 + \frac{2}{\lambda} d^2)}\]
        \For {each $j \in [n]$ s.t. $\sum_{i=1}^t a_{ij} y_i = \mu_j$} \Comment{decrement some variables to maintain dual feasibility}
            \State $m^*_j = \arg \max_{i=1}^t \{a_{ij} | y_i > 0\}$.
            \State Decrease $y_{m^*_j}$ at rate $\frac{a_{tj}}{a_{(m^*_j)j}} \cdot r$.
        \EndFor
    \EndWhile
\end{algorithmic}
\end{algorithm}

On a high level, \Cref{alg:ConvexCoverComplex} follows exactly the same idea as \Cref{alg:ConvexCover}, with only two minor differences: On line 1, $\mu$ is set to $\nabla f(\delta \bar x)$ instead of $\delta \nabla f(\bar x)$\footnote{The choice of $\delta$ here is different between the two algorithms: they are the solutions to two distinct minimization problems.}; On line 12, the numerator of $r$ is $\min_{\ell=1}^n \left\{ \frac{\nabla_\ell f(\delta \bar x)}{\nabla_\ell f(\bar x)} \right\}$ instead of $\delta$. These two changes do not modify the primal update at all, and preserve a vital component in the proof of \Cref{thm:convexmain} that bounds the relative rate of change between $x_j$ and $y_t$ with $\mu_j$, more specifically:
\begin{align*}
\frac{\partial x_j}{\partial y_i} &= \frac{\partial x_j}{\partial \tau} \cdot \frac{\partial \tau}{\partial y_i} = \frac{a_{ij} (x_j + D_j)}{\nabla_j f(x)} \cdot \frac{\log (1 + \frac{2}{\lambda}d^2)}{\min_{\ell=1}^n \left\{ \frac{\nabla_\ell f(\delta \bar x)}{\nabla_\ell f(\bar x)} \right\}}\\
&\geq \log (1 + \frac{2}{\lambda}d^2) \frac{a_{ij} (x_j + D_j)}{\nabla_j f(\delta \bar x)}\\
&= \log (1 + \frac{2}{\lambda}d^2) \frac{a_{ij} (x_j + D_j)}{\mu_j}
\end{align*}
where the inequality follows from the monotonicity of $\nabla f(x)$ and
\begin{align*}
    \min_{\ell=1}^n \left\{ \frac{\nabla_\ell f(\delta \bar x)}{\nabla_\ell f(\bar x)} \right\} \cdot \nabla_j f(x) &\leq \frac{\nabla_j f(\delta \bar x)}{\nabla_j f(\bar x)} \cdot \nabla_j f(x)\\
    &\leq \frac{\nabla_j f(\delta \bar x)}{\nabla_j f(x)} \cdot \nabla_j f(x)\\
    &= \nabla_j f(\delta \bar x)
\end{align*}

The choice of $\mu$ in \Cref{alg:ConvexCoverComplex} no longer allows the bounded growth assumptions of \Cref{fact:boundedgrowth} to simplify the $f^*(\mu)$ term in the dual objective. Thus, the performance of \Cref{alg:ConvexCoverComplex} is stated in the corollary below:

\begin{restatable}{corollary}{ConvexComplex}\label{cor:convexcomplex}
For the learning-augmented online convex covering problem with the monotone gradient assumption, there exists an online algorithm that takes a problem instance and an advice $x'$, and generates a solution $\bar x$ such that 
\[f(\bar x) \leq \min \left\{ O \left( \frac{1}{1-\lambda} \right) f(x') + R(1) \opt, R(\lambda) \opt \right\} \]
where
\begin{align} 
    \frac{1}{R(\lambda)} = \max_{\delta > 0} \min_z &\left( \frac{\min_\ell \left\{ \frac{\nabla_\ell f(\delta z)}{\nabla_\ell f(z)} \right\}}{4 \log (1 + \frac{2}{\lambda} d^2)} - \frac{ \delta z^T \nabla f(\delta z) - f(\delta z)}{f(z)} \right).\label{eq:r-ratio}
\end{align}

Additionally, if $x'$ is feasible, i.e., $A x' \geq \bfone$, then $f(\bar x) \leq O\left(\frac{1}{1-\lambda}\right) f(x')$. Here, $d = \max_{i \in [m]} |\{a_{ij} | a_{ij} > 0\}|$ is the row sparsity of the constraint matrix $A$.
\end{restatable}

While the statement of \Cref{cor:convexcomplex} may be intimidating at first glance, its proof is almost identical to that of \Cref{thm:convexmain}. For completeness' sake, we include a proof in \Cref{app:ConvexComplex}.

%% file: ell-q.tex
\section{Primal-Dual Framework for Online Covering with $\ell_q$-norm Objectives}\label{sec:Lq}

The online covering problem with convex objectives is a very general and expressive framework that encompasses a variety of problems. However, while our algorithmic framework in \Cref{sec:Convex} is general-purpose and can be applied to many such problems, it does require the technical assumption that the gradient of the objective function is monotone, and cannot utilize any inherent structural properties of the problem itself, possibly suffering a loss in generality. 
Is it possible to specialize our algorithms and analyses for online covering with specific convex objective functions? Can we use properties of these objective functions to obtain improved bounds, or remove the monotone gradient assumption?

In this section, we study a more structured variant of the general online convex covering problem: Online covering problems with $\ell_q$-norm objectives. 
We wish to solve the following problem online:
\setcounter{equation}{2}
\begin{align}
  \begin{aligned}
    \text{min } & \sum_{e=1}^r c_e || x (S_e) ||_{q_e}
    \text{ over } x \geq \bfzero 
    \text{ subject to } A x \geq 1.
  \end{aligned}
\end{align}
\setcounter{equation}{5}
where each $S_e \subseteq [n]$ is a subset of indices, $q_e \geq 1$, $c_e \geq 0$, $x(S)$ denote the vector $x$ restricted to indices in $S$, and $||x(S)||_q$ is the $\ell_q$ norm, i.e., $\left( \sum_{i \in S} x_i^q \right)^{1/q}$.

\cite{shen2020online} investigated online covering problems with $\ell_q$-norm objectives~(\ref{eq:lqcovering}), and presented an $O(\log \kappa d)$-competitive algorithm, where $\kappa$ is the condition number, defined as the ratio between the largest and the smallest non-zero entry in $A$, and $d$ is \emph{the maximum between the row sparsity of $A$ and the cardinality of the largest $S_e$}. Their algorithm is in fact identical to a variant of the algorithm of~\cite{buchbinder2014online}, which we base our \Cref{alg:ConvexCover} on, but they pair the algorithm with a new analysis utilizing the structural properties of~(\ref{eq:lqcovering}), removing the assumption in~\cite{buchbinder2014online} that the gradient $\nabla f(x)$ is monotone, which does not hold for $\ell_q$-norms in general.

The Fenchel-based dual that~\cite{shen2020online} considers is the following:
\begin{align}
    \begin{aligned}
        \text{maximize } & \sum_{i=1}^m y_i\\
        \text{over } & y \in \nnreals^m\\
        \text{subject to } & A^T y = \mu\\
        & \sum_{e=1}^r \mu_e = \mu\\
        & ||\mu_e (S_e)||_{p_e} \leq c_e \qquad \forall e \in [r]\\
        & \mu_e (\bar S_e) = \mathbf{0} \qquad \forall e \in [r]
    \end{aligned} \label{eq:lqpacking}
\end{align}
where $p_e$ satisfies that $\frac{1}{p_e} + \frac{1}{q_e} = 1$ for all $e \in [r]$, i.e., $||\cdot||_{p_e}$ is the dual norm of $||\cdot||_{q_e}$. We remark that \eqref{eq:lqpacking} is in fact equivalent to \eqref{eq:convexpacking}, with $f(x) = \sum_{e=1}^r c_e ||x(S_e)||_{q_e}$. As a result, the weak duality of \Cref{fact:duality} still holds for \eqref{eq:lqcovering} and \eqref{eq:lqpacking}. We refer the readers to~\cite{shen2020online} (Section 2) for a detailed derivation of this Fenchel dual.

We present a primal-dual learning-augmented algorithm for online covering problems with $\ell_q$-norm objectives~(\ref{eq:lqcovering}), stated as \Cref{alg:lqCover}, and analyze its performance in \Cref{thm:lqmain}. On a high level, \Cref{alg:lqCover} is identical to \Cref{alg:ConvexCover}, only without the step of decrementing the dual variables for tight packing constraints (Line 15 of \Cref{alg:ConvexCover}). As discussed in \Cref{sec:Convex}, this modification will violate the exact feasibility of the dual solution, incurring an additive $\log \kappa$ factor loss in the competitiveness of the algorithm. However, it preserves the monotonicity of the dual solution, and is vital for our analysis, which removes the reliance on the assumption on the monotonicity of the gradient $\nabla f(x)$.

\begin{algorithm}[!htb]
\caption{Primal-dual learning-augmented algorithm for online covering with $\ell_q$-norm objectives in round $t$} \label{alg:lqCover}

\renewcommand{\algorithmicrequire}{\textbf{Input:}}
\renewcommand{\algorithmicensure}{\textbf{Output:}}

\algorithmicrequire{ $x_1, \cdots, x_n$, $y_1, \cdots, y_{t-1}$, $\mu$: current solution, $A_{tj}$ for $1 \leq j \leq n$: current coefficients, $(x', \lambda)$: advice and confidence parameter.}

\algorithmicensure{ Updated $x$ and $y_t$.}

\begin{algorithmic}[1]
    \State $\tau \gets \text{ current time}$.
    \If {$A_t x' \geq 1$} \Comment{$x'$ is feasible}
        \State $D_j^{(t)} \gets \frac{\lambda}{a_{tj} d} + \frac{(1-\lambda) x'_j \bfone_{x_j < x'_j}}{A_t x'_c}$
    \Else
        \State $D_j^{(t)} \gets \frac{1}{a_{tj} d}$
    \EndIf
    \While {$A_t x < 1$} \label{line:cover}
        \For {each $j \in [n]$ s.t. $a_{tj} > 0$}
            \State Increase $\tau$ at rate $1$. 
            \State Increase $x_j$ at rate
            \[\frac{\partial x_j}{\partial \tau} = \frac{a_{tj}}{\nabla_j f(x)} (x_j + D_j^{(t)})\]
        \EndFor
        \State Increase $y_t$ at rate $1$.
        \State $\mu \gets A^T y$.
    \EndWhile
\end{algorithmic}
\end{algorithm}

\begin{theorem}\label{thm:lqmain}
For the learning-augmented online covering problem with $\ell_q$-norm objectives, there exists an online algorithm that takes a problem instance and an advice $x'$, and generates a solution $\bar x$ such that 
\begin{align*}
    f(\bar x) &\leq \min \left\{ O \left( \frac{1}{1-\lambda} \right) f(x') + O(\log \kappa d) \opt, O\left( \log \frac{\kappa d}{\lambda} \right) \opt \right\}
\end{align*}
Additionally, if $x'$ is feasible, i.e., $A x' \geq \bfone$, then $f(\bar x) \leq O\left(\frac{1}{1-\lambda}\right) f(x')$. Here, $\kappa := \amax / \amin$ is the condition number, $\amax = \max_{i \in [m], j \in [n]} \{a_{ij} | a_{ij} > 0 \}$, $\amin = \min_{i \in [m], j \in [n]} \{a_{ij} | a_{ij} > 0 \}$, and \newline$d = \max \{\max_{i \in [m]} |\{a_{ij} | a_{ij} > 0\}|, \max_{e \in [r]} |S_e| \}$.
\end{theorem}

In our analysis, we assume that all of the index sets $S_e$ are disjoint. This allows us to separate the gradient $\nabla f(x)$ into independent terms for each $e \in [r]$. Recall that $f(x) = \sum_{e=1}^r c_e ||x(S_e)||_{q_e}$. We argue that this assumption is without loss of generality, and that proving \Cref{thm:lqmain} with the disjointness assumption suffices to show that \Cref{alg:lqCover} is consistent and robust even without the disjointness assumption, with the following fact from~\cite{shen2020online}:

\begin{fact}[\cite{shen2020online}]\label{fact:disjoint}
Suppose there is a polynomial-time $O(C)$-consistent $O(R)$-robust algorithm $\mathcal{A}$ for \Cref{eq:lqcovering} for disjoint sets $S_e$, then there is a polynomial-time $O(C)$-consistent $O(R)$-robust algorithm for \Cref{eq:lqcovering} for general instances without the disjointness assumption.
\end{fact}

We outline a proof sketch for \Cref{fact:disjoint} here, and refer the readers to~\cite{shen2020online} for a more detailed proof.

\cite{shen2020online} proves \Cref{fact:disjoint} via a reduction between a general instance of \Cref{eq:lqcovering} and one with equivalent optimal value with disjoint sets. For any instance $\cP$ with general index sets $\{S_e\}_{e \in [r]}$, we can construct an instance $\cP'$ with disjoint sets as follows: for each variable $x_j$ in $\cP$, we create variables $x_j^{(1)}, \cdots, x_j^{(r)}$ in $\cP'$, where $x_j^{(e)}$ corresponds to the possible occurence of $x_j$ in $S_e$. Let $S'_e$ consists of the variables $\{x_j^{(e)}\}_{j\in[n]}$, so that $\{S'_e\}_{e\in[r]}$ is a collection of disjoint sets. For any constraint $\sum_j a_{ij} x_j$ in $\cP$, we create $r^n$ constraints in $\cP'$, each corresponding to a possible combination of the $x_j^{(e)}$ variables, i.e.,
\[\sum_{j=1}^n a_{ij} x_j^{(e_j)} \geq 1 \qquad \forall e_1, \cdots, e_n \in [r]\]
It is simple to show that $\cP$ and $\cP'$ have the same optimal value, and any solution to $\cP'$ can be converted to a solution to $\cP$ by setting $x_j$ to the minimum value over all copies $\min_e \{x_j^{(e)}\}$, maintaining the same objective value. While the reduction isn't a polynomial one due to the exponential amount of constraints in $\cP'$, one can use a separation oracle-based approach as in~\cite{alon2006general}.

While~\cite{shen2020online} only proved the robustness component of \Cref{fact:disjoint}, corresponding to the competitiveness of classical online algorithms without learning-augmentation, we can easily extend their analysis to the consistency as well. Any advice for an instance $\cP$ can be mapped to an advice with equivalent value for $\cP'$ by duplicating each entry $r$ times.

With this disjointness assumption, made without loss of generality, we proceed to prove the robustness and consistency component of \Cref{thm:lqmain} separately.

\paragraph{Robustness.} We first prove the robustness of \Cref{alg:lqCover}, that $f(\bar x) \leq O(\log \frac{\kappa d}{\lambda}) \opt$. While \Cref{alg:lqCover} is similar to \Cref{alg:ConvexCover} on a high level, there are some subtle differences in the detailed specifications. We show the following:
\begin{itemize}
    \item $\bar x$ is feasible;
    \item $(\bar y, \mu)$ is $O(\log \frac{\kappa d}{\lambda})$-approximately feasible;
    \item The primal objective $P$ is at most twice the dual objective $D$.
\end{itemize}
Equipped with these subclaims and weak duality, we have
\[P = f(\bar x) \leq O(1) D \leq O(\log \frac{\kappa d}{\lambda}) \opt \]
as desired.

The feasibility of the primal variables is straightforward, by definition of the algorithm, as with that of \Cref{alg:ConvexCover}. We use the following component lemma to bound the ratio between the primal and the dual objective:
\begin{lemma}
Let $\bar x$ and $\bar y$ be the primal and dual variables at the end of the execution, respectively. The primal objective $f(\bar x)$ is at most twice the dual objective $\sum_{i=1}^m \bar y_i$.
\end{lemma}
\begin{proof}
Consider during round $t$ when the primal constraint $A_t x \geq 1$ arrives. Then,
\begin{align*}
    \frac{\partial f(x^{(\tau)})}{\partial \tau} &= \sum_{j: a_{tj} > 0} \nabla_j f(x^{(\tau)}) \cdot \frac{\partial x_j}{\partial \tau}\\
    &= \sum_{j : a_{tj} > 0} \left( a_{tj} x_j + \frac{\lambda}{d} + \frac{(1-\lambda) a_{tj} x'_j \bfone_{x_j < x'_j}}{A_t x'_c} \right) \\&\leq 1 + \lambda + (1 - \lambda) = 2
\end{align*}
Since the dual objective increases at rate $1$, we conclude as desired.
\end{proof}

The approximate feasibility of the dual variables is the most sophisticated component of the analysis of \Cref{alg:lqCover}. We follow the proof strategy of \cite{shen2020online}. On a high level, in lieu of the assumption that the gradient is monotone, our analysis uses the structural properties of the problem, which the general online convex covering problem lacks, to define a potential function, which we then use to partition the execution of the algorithm into phases, and bound the growth of the dual objective in each phase separately.

\begin{lemma}
The dual solution $\bar y, \bar \mu$ is $O(\log \frac{\kappa d}{\lambda})$ approximately-feasible, i.e.,
\[\mu (\cap_e \bar S_e) = \mathbf{0}\]
\[||\mu (S_e)||_{p_e} \leq c_e \cdot O(\log \frac{\kappa d}{\lambda}) \qquad \forall e \in [r]\]
\end{lemma}
\begin{proof}
The first feasibility condition $\mu (\cap_e \bar S_e) = \mathbf{0}$ follows straightforwardly from the algorithm. Fix any $j \in \cap_e \bar S_e$, and consider during round $t$ when the primal constraint $A_t x \geq 1$ arrives. If $a_{tj} = 0$, then incrementing $y_t$ will not cause $\mu_j$ to increase; If $a_{tj} > 0$, then $\nabla_j f(x) = 0$ will imply that $\frac{\partial x_j}{\partial \tau} = \infty$, so the algorithm will set $x_j$ to the desired value satisfying the constraint without incrementing $y_t$. In either case, $\mu_j$ is not incremented at all, so that at the end of the algorithm, $\mu (\cap_e \bar S_e) = \mathbf{0}$ as desired.

For the second feasibility condition, fix any $e \in [r]$, so that we consider only $f(x) = c_e ||x(S_e)||_{q_e}$. We want to show that $||\mu (S_e)||_{p_e} \leq c_e \cdot O(\log \frac{\kappa d}{\lambda})$. We consider two cases, $q_e = 1$, and the more complicated $q_e > 1$ case.

\paragraph{The case of $q_e = 1$.} The part of the objective function corresponding to our chosen value of $e$ reduces to the linear case, in the form of $f(x) = c_e \sum_{j \in S_e} x_j$, and that $\nabla_j f(x) = c_e$. The corresponding $p_e = \infty$, thus it suffices to show that $||\mu(S_e)||_\infty \leq c_e \cdot O(\log \frac{\kappa d}{\lambda})$, or equivalently, $\mu_j \leq c_e \cdot O(\log \frac{\kappa d}{\lambda})$ for all $j \in S_e$.

First, by definition,
\[\frac{\partial x_j}{\partial \mu_j} = \frac{\partial x_j}{\partial \tau} \frac{\partial \tau}{\partial \mu_j} = \frac{a_{tj} (x_j + D_j^{(t)})}{c_e} \cdot \frac{1}{a_{tj}}\]

Notice that the algorithm never increments any $x_j$ past $\frac{1}{\amin}$, since $a_{tj} x_j \geq \amin x_j$ is at most $1$ during the algorithm. Thus, we can upper bound the total growth of $\mu_j$ over the entire algorithm by
\begin{align*}
\Delta \mu_j &\leq \int_0^{\frac{1}{\amin}} \frac{c_e a_{tj}}{a_{tj} (x_j + D_j^{(t)})} dx_j = c_e \log \left( \frac{\frac{a_{tj}}{\amin}}{a_{tj} D_j^{(t)}} + 1\right) \leq c_e \cdot O(\log \frac{\kappa d}{\lambda})
\end{align*}
as we desire, where the last inequality is due to $\frac{a_{tj}}{\amin} \leq \frac{\amax}{\amin} = \kappa$, and $a_{tj} D_j^{(t)} = \frac{\lambda}{d} + \frac{(1-\lambda) a_{tj} x'_j \bfone_{x_j < x'_j}}{A_t x'_c} \geq \frac{\lambda}{d}$.

\paragraph{The case of $q_e > 1$.} In this case, the objective function corresponding to our chosen value of $e$ is $f(x) = c_e ||x(S_e)||_{q_e} = c_e \left( \sum_{k \in S_e} x_k^{q_e} \right) ^{1/q_e}$. The corresponding partial gradient is 
\begin{align*}
    \nabla_j f(x) &= c_e \cdot \frac{1}{q_e} \left( \sum_{k \in S_e} x_k^{q_e} \right) ^{\frac{1}{q_e} - 1} \cdot q_e x_j^{q_e - 1} = c_e x_j^{q_e - 1} \left( \sum_{k \in S_e} x_k^{q_e} \right) ^{\frac{1}{q_e} - 1}
\end{align*}

We use a potential function to bound the total change of $\mu_j$ throughout the algorithm. With the potential function, we partition the algorithm's execution into phases, and bound the change in each phase separately, before summing the change up across all phases.

We define the potential function as follows:
\[\Phi = \sum_{j \in S_e} x_j^{q_e}\]
Trivially, $\Phi$ is non-decreasing in the primal covering variables $x$. At the beginning of the algorithm, $\Phi = 0$.

Let phase 0 denote the period of the algorithm when $\Phi \leq \zeta := \left(\frac{\lambda}{\amax d^2} \right)^{q_e}$, and for each $\ell \geq 1$, let phase $\ell$ denote the period when $\theta^{\ell-1} \zeta \leq \Phi \leq \theta^\ell \zeta$, for some selection of $\theta$ that will be determined later. Since $x_j$ will never be incremented past $\frac{1}{\amin}$, we know that $\Phi \leq d \left(\frac{1}{\amin}\right)^{q_e}$ at all times, thus there can be at most $3 q_e \log \left( \frac{\kappa d}{\lambda} \right) / \log \theta$ phases.

Note that in \cite{shen2020online}, a similar potential-based strategy is used, but the phases of the algorithm is defined with respect to $\zeta = \left(\frac{1}{\amax d^2} \right)^{q_e}$. The change we make by adding the $\lambda$ factor to the numerator is necessary to suffer an additive error of only $O(\log \frac{1}{\lambda})$, as opposed to $O(\frac{1}{\lambda})$.

\paragraph{Phase 0.} We first bound $\Delta \mu_j$, the change in $\mu_j$, during phase 0. Let $\alpha_j$ denote the value of $x_j$ at the end of phase 0. Then, using a similar strategy to the case of $q_e = 1$, we have
\[\frac{\partial x_j}{\partial \mu_j} = \frac{\partial x_j}{\partial \tau} \frac{\partial \tau}{\partial \mu_j} = \frac{a_{tj} \left(x_j + D_j^{(t)}\right)}{c_e x_j^{q_e-1}\left(\sum_{k \in S_e} x_k^{q_e}\right)^{\frac{1}{q_e}-1}} \cdot \frac{1}{a_{tj}} \]
Rearranging, we obtain
\begin{align*}
    \frac{\partial \mu_j}{\partial x_j} &= \frac{c_e a_{tj} x_j^{q_e-1}}{\left( \sum_{k \in S_e} x_k^{q_e} \right)^{1-\frac{1}{q_e}} a_{tj} \left(x_j + D_j^{(t)}\right)} \leq \frac{d c_e a_{tj} x_j^{q_e-1}}{\lambda \left( \sum_{k \in S_e} x_k^{q_e} \right)^{1-\frac{1}{q_e}} }
\end{align*}
where the last inequality holds since $a_{tj}(x_j + D_j^{(t)}) = a_{tj} x_j + \frac{\lambda}{d} +\frac{(1-\lambda) a_{tj} x'_j \bfone_{x_j < x'_j}}{A_t x'_c} \geq \frac{\lambda}{d}$.

Thus, integrating over the entirety of phase 0, we obtain
\begin{align*}
    \Delta \mu_j &\leq \int_0^{\alpha_j} \frac{d c_e a_{tj} x_j^{q_e-1}}{\lambda \left( \sum_{k \in S_e} x_k^{q_e} \right)^{1-\frac{1}{q_e}} } d x_j\\
     &= \frac{d c_e a_{tj}}{\lambda} \int_0^{\alpha_j} \frac{\left(x_j^{q_e}\right)^{1-\frac{1}{q_e}}}{\left( \sum_{k \in S_e} x_k^{q_e} \right)^{1-\frac{1}{q_e}} } d x_j\\
     &\leq \frac{d c_e a_{tj}}{\lambda} \int_0^{\alpha_j} 1 d x_j = \frac{d c_e a_{tj}}{\lambda} \alpha_j
\end{align*}

Since in phase 0, we have $\Phi = \sum_{k \in S_e} x_k^{q_e} \leq \left( \frac{\lambda}{\amax d^2} \right)^{q_e}$, it must hold that $a_j \leq \frac{\lambda}{\amax d^2}$. Thus
\[\Delta \mu_j \leq \frac{d c_e a_{tj}}{\lambda} \alpha_j \leq \frac{c_e}{d} \]

Taking the $p_e$ norm over all $j \in S_e$, we have that at the end of phase 0, $||\mu(S_e)||_{p_e} \leq ||\mu(S_e)||_1 \leq c_e$, since $d \geq \max_e |S_e|$ by definition.

\paragraph{Phase $\ell \geq 1$.} We then bound $\Delta \mu_j$ during some phase $\ell \geq 1$. Denote by $s_j$ and $t_j$ the value of $x_j$ at the beginning and the end of the phase, and by $\Phi_0$ and $\Phi_1$ the value of $\Phi$ at the beginning and the end, respectively.

Similarly, bounding the rate of change of $\mu_j$ with that of $x_j$ yields
\begin{align}
\frac{\partial \mu_j}{\partial x_j} &= \frac{c_e a_{tj} x_j^{q_e-1}}{\left( \sum_{k \in S_e} x_k^{q_e} \right)^{1-\frac{1}{q_e}} a_{tj} \left(x_j + D_j^{(t)}\right)} \leq \frac{c_e x_j^{q_e-2}}{\left( \sum_{k \in S_e} x_k^{q_e} \right)^{1-\frac{1}{q_e}} }\label{eq:phase1partial}
\end{align}
where the last inequality holds since $a_{tj}(x_j + D_j^{(t)}) = a_{tj} x_j + \frac{\lambda}{d} +\frac{(1-\lambda) a_{tj} x'_j \bfone_{x_j < x'_j}}{A_t x'_c} \geq a_{tj} x_j$.

Bounding the right-hand side of \Cref{eq:phase1partial} with $\frac{c_e}{x_j}$ as per our strategy in phase 0, and integrating over the range of $x_j$ yields $\Delta \mu_j \leq c_e \log (t_j/s_j)$, which depends on the value of $s_j$ and $t_j$ and is not tight enough. Instead, we use our defined potential $\Phi$, and the fact that $\Phi = \sum_{k \in S_e} x_k^{q_e} \geq \Phi_0$ throughout this phase, to obtain the following bound:
\begin{equation}\label{eq:phase1delta}
\Delta \mu_j \leq \int_{s_j}^{t_j} \frac{c_e x_j^{q_e-2}}{\Phi_0^{1-\frac{1}{q_e}}} = \frac{c_e (t_j^{q_e-1} - s_j^{q_e-1})}{(q_e-1)\Phi_0^{1-\frac{1}{q_e}}}
\end{equation}
Here, our assumption that $q_e > 1$ is required to evaluate the integral.

Towards bounding the $p_e$ norm, we have
\[(\Delta \mu_j)^{p_e} \leq \frac{c_e^{p_e}(t_j^{q_e-1} - s_j^{q_e-1})^{p_e}}{(q_e-1)^{p_e}\Phi_0^{p_e(1-\frac{1}{q_e})}} \]
We use two mathematical facts to simplify \Cref{eq:phase1delta}. The first is the fact that $a^{p_e} + b^{p_e} \leq (a+b)^{p_e}$ for any $a, b \geq 0$ and $p_e \geq 1$. Setting $a = (t_j^{q_e-1} - s_j^{q_e-1})$ and $b = s_j^{q_e-1}$ yields $(t_j^{q_e-1} - s_j^{q_e-1})^{p_e} \leq t_j^{(q_e-1)p_e} - s_j^{(q_e-1)p_e}$.

The second fact we use is that by definition, $\frac{1}{p_e} + \frac{1}{q_e} = 1$. We have both $p_e (1 - \frac{1}{q_e}) = 1$, and $(q_e - 1) p_e = q_e$.

Thus, we can simplify \Cref{eq:phase1delta} to obtain
\[(\Delta \mu_j)^{p_e} \leq \frac{c_e^{p_e}(t_j^{q_e-1} - s_j^{q_e-1})^{p_e}}{(q_e-1)^{p_e}\Phi_0^{p_e(1-\frac{1}{q_e})}} \leq \frac{c_e^{p_e}(t_j^{q_e} - s_j^{q_e})}{(q_e-1)^{p_e}\Phi_0}\]

Finally, summing over all indices $j \in S_e$, we have
\begin{align*}
\sum_{j \in S_e} (\Delta \mu_j)^{p_e} &\leq \sum_{j \in S_e} \frac{c_e^{p_e}(t_j^{q_e} - s_j^{q_e})}{(q_e-1)^{p_e}\Phi_0}\\
&= \frac{c_e^{p_e}}{(q_e-1)^{p_e} \Phi_0} \left( \sum_{j \in S_e} t_j^{q_e} - \sum_{j \in S_e} s_j^{q_e} \right)\\
&= \frac{c_e^{p_e}}{(q_e-1)^{p_e} \Phi_0} (\Phi_1 - \Phi_0)\\ 
&\leq \frac{c_e^{p_e}(\theta - 1)}{(q_e-1)^{p_e}}
\end{align*}
where the equality on the third line is by definition of $\Phi$ and selection of $s_j, t_j, \Phi_0, \Phi_1$; and the inequality on the fourth line is by our definition of phases, that $\Phi_0 = \theta^{\ell - 1} \zeta$ and $\Phi_1 = \theta^{\ell} \zeta$.

\paragraph{Combining over all phases.} Let $\mu^{(\ell)}$ denote the increment of the vector $\mu$ during phase $\ell$, and let $L$ denote the total amount of phases. Then $\mu = \sum_{\ell = 0}^L \mu^{(\ell)}$. By the triangle inequality for $\ell_q$\new{-}norms, we have
\begin{align*}
    ||\mu||_{p_e} &\leq \sum_{\ell=0}^L ||\mu^{(\ell)}||_{p_e}\\
    &\leq c_e + \left(\frac{c_e^{p_e} (\theta - 1)}{(q_e-1)^{p_e}}\right)^{\frac{1}{p_e}} \cdot 3 q_e \frac{\log \frac{\kappa d}{\lambda} }{\log \theta}\\
    &= c_e \cdot \left(1 + \frac{3 q_e (\theta - 1)^{1/p_e}}{(q_e - 1) \log \theta} \cdot \log \frac{\kappa d}{\lambda} \right)
\end{align*}
by our individual bounds for each phase, and the fact that $L \leq 3 q_e \log(\frac{\kappa d}{\lambda}) / \log \theta$.

\paragraph{Choosing the value of $\theta$.} It remains to choose an appropriate value of $\theta$ for any $q_e > 1$, such that $\frac{3 q_e (\theta - 1)^{1/p_e}}{(q_e - 1) \log \theta} \leq O(1)$, or equivalently, $||\mu||_{p_e} \leq c_e \cdot O(\log \frac{\kappa d}{\lambda})$, as we desire.

For $q_e \geq 2$, we can simply choose $\theta = 2$, which gives us 
\[\frac{3 q_e (\theta - 1)^{1/p_e}}{(q_e - 1) \log \theta} \leq 6\]
as desired.

For $1 < q_e < 2$, we choose $\theta = 1 + (q_e - 1)^{- \eps p_e}$, where $\eps = \frac{1}{- \log (q_e - 1)} > 1$. Then
\begin{align*}
\frac{3 q_e (\theta - 1)^{1/p_e}}{(q_e - 1) \log \theta} &\leq \frac{3 q_e (\theta - 1)^{1/p_e}}{(q_e - 1) \log (q_e - 1)^{- \eps p_e}} \qquad \text{since $\theta \geq (q_e - 1)^{- \eps p_e}$}\\
&= \frac{3 q_e}{q_e - 1} \cdot \frac{(q_e-1)^{-\eps}}{- \eps p_e \log (q_e - 1)} \\
&= \frac{3 q_e}{q_e - 1} \cdot \frac{e}{p_e}\\
&\leq 9 \qquad \text{since $p_e (q_e-1) = q_e$}
\end{align*}
as desired.

Thus, for all values of $q_e > 1$, we have that
\[||\mu||_{p_e} \leq c_e \cdot \left( 1 + 9 \log \frac{\kappa d}{\lambda} \right) = c_e \cdot O(\log \frac{\kappa d}{\lambda}) \]
which concludes the proof.
\end{proof}

We remark that the component of our above proof for phase $\ell \geq 1$ is in fact identical to that of \cite{shen2020online}, and that introducing learning-augmentation with the confidence parameter $\lambda$ does not modify the proof explicitly. However, since we defined $\zeta$ as $\left(\frac{\lambda}{\amax d^2}\right)^{q_e}$, our analysis sums over a different amount of total phases, dependent on $\lambda$.

\paragraph{Consistency.} We remark that \Cref{lem:convexconsistency} holds independent of the dual update theme and the monotone gradient assumption. Thus, \Cref{lem:convexconsistency} still holds for \Cref{alg:lqCover}.

%% file: applications.tex
\section{Applications} \label{sec:Appl}

\subsection{Applications of Concave Packing}\label{sec:ApplPacking}

In this section, we present several direct applications of general interest of our switching-based algorithmic framework where \Cref{thm:packingmain} immediately gives non-trivial results, in some cases even optimal, up to constant factors.

\subsubsection{Online Packing Linear Programming} \label{sec:packing-lp}

The online packing linear programming (LP) problem is a special case of \eqref{eq:concavepacking} where the objective $g(y)$ is the linear function $\sum_{i=1}^m y_i = \mathbf{1}^T y$.
\begin{align}
  \begin{aligned}
    \text{max } &  \mathbf{1}^T y
    \text{ over } y \in \nnreals^m 
    \text{ subject to } A^T y \leq b.
  \end{aligned} \label{eq:packing-lp}
\end{align}

The seminal work of~\cite{buchbinder2009online} presents the following theorem.

\begin{theorem}[\cite{buchbinder2009online}] \label{thm:BN09}
For any $B > 0$, there exists a $B$-competitive algorithm for online packing LP \eqref{eq:packing-lp} such that for each constraint $j \in [n]$ it holds:
\[ \sum_{i=1}^m a_{ij} y_i = b_j \cdot O \left( \frac{\log n \kappa_j}{B} \right), \]
where $\kappa_j=a_j(\max)/a_j(\min)$, $a_j(\max) = \max_{i=1}^m \{a_{ij}\}$, and $a_j(\min) = \min_{i=1}^m \{a_{ij} | a_{ij} > 0\}$.
\end{theorem}
For impossibility results, it is further shown in~\cite{buchbinder2009online} that for any $B > 0$, there exists a single-constraint instance for the online packing LP problem, such that any $B$-competitive algorithm must violate the constraint by a factor of $\omega(\log n \kappa / B)$.

The following corollary holds by Theorems \ref{thm:packingmain} and \ref{thm:BN09}. Here, $\kappa:=\max_{j=1}^n\{\kappa_j\}$.

\begin{corollary} \label{cor:packing-lp}
    There exists a $\frac{1}{1-\lambda}$-consistent, $\frac{B}{\lambda}$-robust, and $O(\log n \kappa/B)$-feasible learning-augmented online algorithm for online packing LP \eqref{eq:packing-lp} that takes parameter $B>0$, $\lambda \in [0, 1]$, and an advice which is $O(\log n \kappa/B)$-feasible.
\end{corollary}
\Cref{cor:packing-lp} is optimal in consistency, robustness, and feasibility, up to constant factors. We note that the parameter $\kappa$ is non-decreasing over time. It is possible that Algorithm \ref{alg:simple} enters line \ref{line:alg-simple-case-2} in earlier rounds and then enters line \ref{line:alg-simple-case-1} in later rounds when $\beta=\Theta(\log n \kappa/B)$ increases.

In the following subsections, we present other specific applications for learning-augmented online packing LP.

\subsubsection{Online Knapsack} \label{sec:knapsack}

In this section, we apply Algorithm~\ref{alg:simple} to the classical online \emph{knapsack} problem. In the most basic formulation of this problem, a sequence of $m$ items arrive online, and each item $i$ has value $v_i$ and weight $w_i$. Our goal is to pack a selected set of items with as much value as possible into a knapsack of capacity $C$, deciding irrevocably whether to include each item as soon as it arrives online, without the total weight exceeding the capacity. Formulated as an online packing linear program, we obtain
\begin{equation*}
    \text{max } \mathbf{1}^T y 
    \text{ over } y \in \nnreals^m 
    \text{ subject to } \sum_{i=1}^m \frac{w_i y_i}{v_i} \leq C. \label{equation:knapsack}
\end{equation*}
in which the elements $y_i$ together with it's value $v_i$ and weight $w_i$ arrive online.

As a corollary to \Cref{thm:packingmain}, it follows that our simple algorithm can be applied to the knapsack problem to obtain asymptotically optimal performance in both consistency and robustness.

\begin{corollary}
Given any $\alpha$-competitive algorithm $\mathcal{O}$ for the online knapsack problem, \Cref{alg:simple} using $\mathcal{O}$ as a subroutine implies a $\frac{1}{1-\lambda}$-consistent with $1$-feasible advice, $\frac{\alpha}{\lambda}$-robust, and $(2-\lambda)$-feasible learning-augmented online algorithm.
\end{corollary}

The knapsack problem is regarded as one of the most fundamental packing problems in computer science and operations research along with its many variants. It has been extensively studied in the offline setting~\cite{martello1987algorithms} (see~\cite{salkin1975knapsack} for a survey)  
and the online setting~\cite{marchetti1995stochastic,ma2019competitive}. Recent work in the learning-augmented setting includes~\cite{daneshvaramoli2023online,IKMP2021,boyar2022online}, which study the online knapsack problems with learning augmentation from many angles, such as frequency predictions, threshold predictions, and unit profit settings. 

We point out that our advice model on the values of the arriving dual variables generalizes most, if not all,  forms of advice studied in prior literature. For any advice suggesting a course of action, such as taking an arriving item into the knapsack (or doing so with some probability, in the case of randomized algorithms), we can process such advice by setting the corresponding entry in $y'$ to its value (resp. some fraction of its value).

\subsubsection{Online Resource Management Benefit}

In this section, we apply \Cref{alg:simple} to the online resource management benefit problem~\cite{leonardi1995line}. For this problem, we have $n$ distinct resources. Each resource $j \in [n]$ can be assigned a maximum capacity $c_j$. Here, $n$ and each $c_j$ are given in advance. A sequence of $m$ jobs is presented online, one job at a time, where $m$ can be unknown. Job $i \in [m]$ has a benefit of $w_i$ and can be scheduled with $r_i$ different alternatives. We use $a^{(k)}_{ij}$ for $k \in [r_i]$ to denote the amount of resource $j$ necessary to schedule job $i$ with alternative $k$. Upon the arrival of job $i$, $w_i$, $r_i$, and $a^{(k)}_{ij}$ are revealed, and one must schedule job $i$ to different alternatives irrevocably. This problem has the following LP formulation:
\begin{equation*} \label{equation:max-benefit}
\begin{aligned}
& \max_{y} & & \sum_{i=1}^m \sum_{k=1}^{r_i} y^{(k)}_i \\
& \text{subject to}
& & \sum_{i=1}^m \sum_{k=1}^{r_i} \frac{a^{(k)}_{ij} y^{(k)}_i}{w_i} \leq c_j & \forall j \in [n],\\
& & & \sum_{k=1}^{r_i} \frac{y^{(k)}_i}{w_i} \le 1 & \forall i \in [m],\\
& & & y^{(k)}_i \ge 0 & \forall i \in [m], k \in [r_i].
\end{aligned}
\end{equation*}

Here, $y^{(k)}_i$ denotes the fraction of job $i$ assigned to alternative $k$. Upon the arrival of job $i \in [m]$, the decision $y^{(k)}_i$ is irrevocable. 

In \cite{leonardi1995line}, it is assumed that the benefit of each job $w_i \in [1,W]$ with $W \ge 1$ and for every non-zero coefficient $a^{(k)}_{ij}$, the ratio $a^{(k)}_{ij}/c_j \in [1/P,1]$ with $P \ge 1$. \cite{leonardi1995line} presents an $O(\log nWP)$-competitive algorithm and shows it is the best possible. As a corollary, it follows that with advice, we can achieve the best outcome either obtained from an online algorithm or the advice by losing a constant factor, in a black box manner.

\begin{corollary}
Given any $\alpha$-competitive algorithm $\mathcal{O}$ for the online resource management problem, \Cref{alg:simple} using $\mathcal{O}$ as a subroutine implies a $\frac{1}{1-\lambda}$-consistent with $1$-feasible advice, $\frac{\alpha}{\lambda}$-robust, and $(2-\lambda)$-feasible learning-augmented online algorithm.
\end{corollary}

\subsubsection{Online Throughput Maximization}

In the Online Throughput Maximization problem~\cite{buchbinder2006improved,awerbuch1993throughput}, we are given a directed or undirected graph $G=(V,E)$. Each edge has a capacity $c: E \to \nnreals$. A set of requests $(s_i,t_i) \in V \times V$ arrive online, one at a time. Upon the arrival of request $i$, the algorithm irrevocably chooses (fractionally) an $s_i$-$t_i$ path and allocates a total bandwidth of one unit. The objective is to maximize the throughput. The following packing LP describes this problem.
\begin{equation*}
\begin{aligned}
& \max_{y} & & \sum_{i} \sum_{P \in \cP_i}y_{i,P} \\
& \text{subject to}
& & \sum_{P \in \cP_i} y_{i,P} \le 1 & \forall i,\\
& & & \sum_{i} \sum_{P \in \cP_i: e \in P} y_{i,P} \le c(e) & \forall e \in E,\\
& & & y_{i,P} \le 0 & \forall i \quad \forall P \in \cP_i,\\
\end{aligned}
\end{equation*}
where $\cP_i$ denotes the set of all the $s_i$-$t_i$ paths. Here, $y_{i,P}$ indicates the amount of flow on the path $P$. The first set of constraints captures that each request requires a unit of bandwidth. The second set of constraints captures that the load on each edge does not exceed the edge capacity.

An $O(1)$-competitive $O(\log |V|)$-feasible online algorithm is presented in~\cite{buchbinder2006improved}, which together with our result implies the following.

\begin{corollary}
    There exists a $\frac{1}{1-\lambda}$-consistent with $O(\log |V|)$-feasible advice, $O(\frac{1}{\lambda})$-robust, and $O(\log |V|)$-feasible learning-augmented algorithm for the online throughput maximization problem.
\end{corollary}

\subsubsection{Online Network Utility Maximization}

In the online network utility maximization (ONUM) problem~\cite{cao2022online}, we are given a directed or undirected graph $G=(V,E)$. Each edge has a capacity of one unit. A set of requests $(s_i,t_i) \in V \times V$ arrive online one at a time. Each request $i$ is associated with a budget $b_i \in \R_{>0}$, a monotone concave utility function $g_i: \nnreals \to \nnreals$ with $g(0) = 0$, and a specified $s_i$-$t_i$ path $P_i$. Upon the arrival of request $i$, one must decide irrevocably its allocation rate without exceeding the budget $b_i$. The goal is to approximately maximize the total utility under the online allocation requirement. This problem has the following packing formulation.
\begin{equation*} 
\begin{aligned}
& \max_{y} & & \sum_i g(y_i) \\
& \text{subject to}
& & \sum_{i: e \in P_i} y_i \leq 1 & \forall e \in E,\\
& & & y_i \in [0,b_i] & \forall i.\\
\end{aligned}
\end{equation*}
In~\cite{cao2022online}, a sufficient condition of having a roughly linear (in $|E|$) competitive algorithm for ONUM is presented. Together with our results, this implies the following corollary.
\begin{corollary}
Given any $\alpha$-competitive algorithm $\mathcal{O}$ for ONUM, \Cref{alg:simple} using $\mathcal{O}$ as a subroutine implies a $\frac{1}{1-\lambda}$-consistent with $1$-feasible advice, $\frac{\alpha}{\lambda}$-robust, and $(2-\lambda)$-feasible learning-augmented online algorithm.
\end{corollary}

\subsubsection{Online Optimization under Inventory Constraints}

The online optimization with inventory constraints (OOIC) problem is introduced in~\cite{lin2019competitive}. In this problem, a decision maker sells inventory across an interval of $T$ discrete time slots to maximize the total revenue. At time $t \in [T]$, the decision maker observes the revenue function $g_t: \nnreals \to \nnreals$ and makes an irrevocable decision on the quantity $y_t$. Upon choosing $y_t$, the decision maker receives revenue $g_t(y_t)$. The goal is to approximately maximize the total revenue, while respective the inventory constraint $\sum_{t \in T}y_t \le \Delta$ with a given parameter $\Delta > 0$. It is further assumed that for all $t \in [T]$, $g_t$ has the following nice properties: (1) $g_t$ is concave, increasing, and differentiable over $[0,\Delta]$, (2) $g_t(0) = 0$, and (3) $g_t'(0) > 0$ and $g'(0) \in [m, M]$ for some $M \ge m > 0$. OOIC has the following packing formulation:
\begin{equation*} 
\begin{aligned}
& \max_{y} & & \sum_{t \in [T]} g_t(y_t) \\
& \text{subject to}
& & \sum_{t \in [T]} y_t \leq \Delta,\\
& & & y_t \ge 0 & \forall t \in [T].\\
\end{aligned}
\end{equation*}
\cite{lin2019competitive} presented a tight $\pi$-competitive algorithm for OOIC, where $\pi \in [\log(M/m)+1, \eta(\log(M/m)+1)]$ and $\eta:=\sup_{g,x \in [0,\Delta]} \{g'(0)x/g(x)\}$. Together with our results, this implies the following corollary.
\begin{corollary}
Given any $\alpha$-competitive algorithm $\mathcal{O}$ for OOIC, Algorithm~\ref{alg:simple} using $\mathcal{O}$ as a subroutine implies a $\frac{1}{1-\lambda}$-consistent with $1$-feasible advice, $\frac{\alpha}{\lambda}$-robust, and $(2-\lambda)$-feasible learning-augmented online algorithm.
\end{corollary}

\subsection{Application of Convex Covering: Online Mixed Covering and Packing}\label{sec:ApplCovering}

In this section, we present an application of our PDLA framework for online convex covering problems to online mixed covering and packing. We remark that our application can be employed to design learning-augmented online algorithms for a vast list of problems that can be formulated as mixed covering and packing problems, including online capacity-constrained facility location, online capacitated multicast, and different variants of online set cover problems. We refer the reader to \cite{buchbinder2014online} for a more detailed description of the applications to these problems.

In the mixed covering and packing problem, we have $k$ linear objectives. Let $x \in \R^n_{\ge 0}$ be the decision variables. The vector of the linear objectives is captured by $B x$ for a matrix $B \in \R^{k \times n}_{\ge 0}$. The goal is to minimize the $\ell_q$-norm of the linear objectives subject to the covering constraint $A x \ge 0$ where $q \ge 1$ and $A \in \R^{m \times n}_{\ge 0}$. This problem can be formulated as
\begin{align}
  \begin{aligned}
    \text{minimize } & || Bx ||_q
    \text{ over } x \in \nnreals^n 
    \text{ subject to } A x \geq \bfone.
  \end{aligned} \label{eq:mixed}
\end{align}
In the online problem, the matrix $B$ is given offline, and the rows of $A$ arrive one at a time where $m$ is unknown. The goal is to update $x$ in a non-decreasing manner to satisfy each arriving constraint and approximately minimize the objective. In the learning-augmented problem, we are also given the advice $x'$. Let $\opt$ denote the optimal objective value of \eqref{eq:mixed}. Corollary \ref{cor:convexcomplex} implies the following.

\begin{corollary}\label{cor:mixed}
For the learning-augmented online mixed covering and packing problem \eqref{eq:mixed}, there exists an online algorithm that takes a problem instance and an advice $x'$, and generates a solution $\bar x$ such that 
\begin{align*}
    ||B \bar x||_q \leq \min &\left\{ O \left( \frac{1}{1-\lambda} \right) ||B x'||_q + O(q \log d) \opt, O\left(q \log \frac{d}{\lambda} \right) \opt \right\}
\end{align*}
Additionally, if $x'$ is feasible, i.e., $A x' \geq \bfone$, then $||B \bar x||_q \leq O\left(\frac{1}{1-\lambda}\right) ||B x'||_q$. Here, $d = \max_{i \in [m]} |\{a_{ij} | a_{ij} > 0\}|$ is the row sparsity of the constraint matrix $A$.
\end{corollary}

\begin{proof}
    Let $f(x) = ||B x||^q_q$.\footnote{One might ask if we can use $|| Bx ||_q$ as the objective. The gradient of $|| Bx ||_q$ is not monotone, but we need the assumption that the gradient is monotone in \eqref{eq:dxdy-mono-grad} for the analysis to hold.} We consider the optimization problem
    \begin{align}
  \begin{aligned}
    \text{min } & f(x)
    \text{ over } x \in \nnreals^n 
    \text{ subject to } A x \geq \bfone \label{opt:mixed-q}
  \end{aligned}
\end{align}
and then take the $q$-th root of the objective.
    
    Let $b_{ij}$ denote the $i$-th row $j$-th entry of $B$ and $B_i$ denote the $i$-th row of $B$, we have that
    \[f(z) = \sum_{i=1}^k \left(B_i z\right)^q, f(\delta z) = \delta^q \sum_{i=1}^k \left(B_i z\right)^q,\]
    \[
        \nabla_j f(z) 
        = q\left(\sum_{i=1}^k b_{ij} \left(B_i z\right)^{q-1}\right),\] 
and
    \[ \nabla_j f(\delta z) = q \delta^{q-1} \left(\sum_{i=1}^k b_{ij} \left(B_i z\right)^{q-1}\right).\]
    This implies that $\nabla_j f(\delta z) / \nabla_j f(z) = \delta^{q-1}$ for all $j \in [n]$, $f(\delta z) = \delta^q f(z)$, and thus
    \begin{align*}
        z^T \nabla f(\delta z) &=  q \delta^{q-1} \sum_{i=1}^k \sum_{j=1}^n b_{ij} z_j \left(B_i  z\right)^{q-1} = q \delta^{q-1} \sum_{i=1}^k \left(B_i z\right)^q = q \delta^{q-1} f(z)
    \end{align*}
    for all $z \in \R_{\ge 0}$. We use Algorithm \ref{alg:ConvexCoverComplex} with \eqref{eq:r-ratio} and have that
\begin{align*}
    \frac{1}{R(\lambda)} &= \max_{\delta > 0} \min_z \left( \frac{\min_\ell \left\{ \frac{\nabla_\ell f(\delta z)}{\nabla_\ell f(z)} \right\}}{4 \log (1 + \frac{2}{\lambda} d^2)} - \frac{ \delta z^T \nabla f(\delta z) - f(\delta z)}{f(z)} \right) \\
    &= \max_{\delta > 0}\left( \frac{\delta^{q-1}}{4 \log (1 + \frac{2}{\lambda} d^2)} - q \delta^q + \delta^q \right)\\
    &= \frac{1}{\left(4q \log (1 + \frac{2}{\lambda} d^2)\right)^q}
\end{align*}
when $\delta = 1/(q\log (1 + 2 d^2/\lambda))$. Let $\opt^q$ be the optimal objective value for \eqref{opt:mixed-q}. We have that
\begin{align*}
    ||B \bar x||^q_q &\leq \min \left\{ O \left( \frac{1}{1-\lambda} \right) ||B x'||^q_q  + O(\left(q \log d\right)^q) \opt^q, O\left(\left(q \log \frac{d}{\lambda}\right)^q \right) \opt^q \right\}
\end{align*}

which implies
\begin{align*}
  O \left( \frac{1}{1-\lambda} \right) ||B x'||^q_q + O(q \log d) \opt^q &\le O \left( \left(\frac{1}{1-\lambda}\right)^q \right) ||B x'||^q_q + O(q \log d) \opt^q \\
  &\le \left( O \left( \frac{1}{1-\lambda} \right) ||B x'||_q + O(q \log d) \opt \right)^q,
\end{align*}
thus,
\begin{align*}
    ||B \bar x||_q \leq \min &\left\{ O \left( \frac{1}{1-\lambda} \right) ||B x'||_q + O(q \log d) \opt,  O\left(q \log \frac{d}{\lambda} \right) \opt \right\}
\end{align*}
as required.
\end{proof}

\subsection{Applications of $\ell_q$-Norm Covering: Online Buy-at-Bulk Network Design}\label{sec:ApplLq}

In this section, we present an application of our PDLA framework for online covering problems with $\ell_q$-norm objectives to the online buy-at-bulk network design problem.

In the buy-at-bulk network design problem, we are given a directed or undirected $n$-vertex graph $G=(V,E)$. Each edge $e \in E$ is associated with a monotone subadditive cost function $g_e: \R_{\ge 0} \to \R_{\ge 0}$ with $g_e(0) = 0$. We are also given a set of $k$ terminal pairs $\{(s_i, t_i) \mid i \in [k]\} \subseteq V \times V$. The goal is to find a collection of $s_i$-$t_i$ paths $P_i$ for each $i \in [k]$ such that the overall cost $\sum_{e \in E}g_e(load_e)$ is minimized, where $load_e$ is the number of paths using $e$. By paying an approximation factor of two, the objective can be written in terms of the \emph{upfront cost} $c_e$ and \emph{pay-per-use cost} $\ell_e$, i.e.,
$\sum_{e \in \cup P_i} c_e + \sum_{e \in E} \ell_e \cdot load_e$. In the online setting, the graph and the cost functions are given in advance. The terminal pairs arrive one at a time. The goal is to irrevocably assign an $s_i$-$t_i$ path upon the arrival of pair $i$ while approximately minimizing the overall cost. In the learning-augmented setting, upon the arrival of pair $i$, we are also given a $s_i$-$t_i$ path as advice.

The seminal framework for online buy-at-bulk in \cite{cekp} gave a modular online algorithm with competitive ratio $O(\alpha \beta \gamma \log^5 n)$. The ratio was later improved by \cite{shen2020online} to $O(\alpha \beta \gamma \log^3 n)$. Here,
\begin{itemize}
    \item $\alpha$ denotes the \emph{junction tree} approximation ratio. A junction tree rooted at a vertex $r \in V$ is a union of $s_i$-$t_i$ paths where each path passes through $r$. A junction tree solution is a collection of junction trees rooted at different vertices, such that each terminal pair $(s_i,t_i)$ is connected by using one of the junction trees with a specified root $r \in V$. The cost of a junction tree solution is not shared across junction trees with different roots. The value $\alpha$ is the worst-case ratio between the cost of an optimal junction tree solution and the optimum.
    \item $\beta$ is the integrality gap of the natural LP relaxation for single-source buy-at-bulk instances.
    \item $\gamma$ is the competitive ratio of an online algorithm for single-source buy-at-bulk instances.
\end{itemize}
We now describe the convex optimization formulation used in \cite{shen2020online}. Let $T = \{s_i,t_i \mid i \in [k]\}$ denote the set of terminal vertices. Let $x_{r,u,e}$ denote the \emph{flow} on edge $e \in E$ connecting $u$ and $r$ in a junction tree rooted at $r \in V$. For any $S \subseteq E$, $r \in V$, and $u \in T$, we define $x_{r,u}(S):= \sum_{e \in S}x_{r,u,e}$. This notion is useful when $S$ forms a $u$-$r$ cut while considering a junction tree rooted at $r$. The optimization problem is as follows.
\begin{equation} \label{opt:bab}
\begin{aligned}
& \min_{x} & & \sum_{r \in V}\sum_{e \in E}c_e \max_{u \in T}\left\{x_{r,u,e}\right\} + \sum_{r \in V} \sum_{e \in E} \ell_e \sum_{u \in T} x_{r,u,e}\\
& \text{s.t.}
& & \sum_{r \in R_s} x_{r,s_i}(S_r) + \sum_{r \in R_t} x_{r,t_i}(T_r) \ge 1 \\
& & & \qquad \qquad \forall i \in [k], \quad \forall (R_s,R_t) \text{ partition of }V, \\
& & & \qquad \qquad \forall S_r: s_i\text{-}r \text{ cut}, \forall r \in R_s, \\
& & & \qquad \qquad \forall T_r: t_i\text{-}r \text{ cut}, \forall r \in R_t,\\
& & & \; x_{r,u,e} \geq 0 \qquad\qquad \forall e \in E, \quad \forall u \in T, \quad \forall r \in V.\\
\end{aligned}
\end{equation}
Consider a feasible $0$-$1$ solution $f$. In the objective, for each junction tree rooted at $r \in V$, the upfront cost paid for using edge $e \in E$ is $c_e \max_{u \in T}\left\{x_{r,u,e}\right\}$ since $x_{r,u,e}$ would indicate if $e$ is used in the junction tree connecting $r$ and $u$; the pay-per-use cost is $\ell_e \sum_{u \in T} x_{r,u,e}$ since $\sum_{u \in T} x_{r,u,e}$ counts the number of paths using $e$. The first constraint says that $s_i$ and $t_i$ must be fractionally connected to the roots of different junction trees with a total flow value of at least one.

An equivalent LP formulation for \eqref{opt:bab} can be solved online by paying a factor of $O(\log^3 n)$. The improvement in \cite{shen2020online} is done as follows. The term $\max_{u \in T}\left\{x_{r,u,e}\right\}$ can be written as the $\ell_{\infty}$-norm of $x$. Note that $\ell_{\log n}$-norm is a constant approximation for $\ell_{\infty}$-norm. The objective can be reformulated as the sum of $\ell_{\log n}$-norm and $\ell_1$-norm and the constraint is $0$-$1$. This is in the form of \eqref{eq:lqcovering} with $\kappa = 1$, which implies an $O(\log n)$ competitive algorithm as presented in \cite{shen2020online}.

Now we consider the learning-augmented setting. Recall that upon the arrival of pair $i$, we are also given an $s_i$-$t_i$ path as advice. Consider the solution constructed by taking the union of the $s_j$-$t_j$ paths for $j \in [i]$. We convert the advice solution into a junction tree solution by paying a factor of $\alpha$ using previous frameworks~\cite{antonakopoulos2010approximating,chekuri2010approximation}. This introduces advice for the $0$-$1$ flow variables termed as $x'$. Let $\adv$ denote the cost of the advice. We have that $f(x') \le \alpha \adv$. Using Theorem \ref{thm:lqmain} and the rounding scheme for \eqref{opt:bab} in \cite{cekp} which pays a factor of $O(\log^2 n)$, we have the following.

\begin{corollary}\label{cor:bab}
For the learning-augmented online buy-at-bulk problem, there exists an online algorithm that takes a problem instance and advice consisting of $s_i$-$t_i$ paths and generates an online solution with cost at most
\[\min \left\{ O \left( \frac{\alpha}{1-\lambda} \right) \adv, O\left(\alpha \beta \gamma \log^2 n \log \frac{n}{\lambda} \right) \opt \right\}.\]
\end{corollary}

We note that using a minimum cut algorithm as a separation oracle, the learning-augmented algorithm runs in polynomial time. Here we state the high-level idea that slightly modifies Algorithm \ref{alg:lqCover}. In the line \ref{line:cover} while loop, whenever $A_t x < 1/2$, we increment $x$ until $A_t x \ge 1$. The while loop is not entered if $A_t x \ge 1/2$. This guarantees that the growth of $x$ is sufficiently large so that the number of iterations entering the line \ref{line:cover} while loop is $O(E)$. We return $2x$ as the final solution. Since the objective function $f$ is convex and $f(\bfzero)=0$, the solution is feasible and we pay a constant factor for the objective.

%% file: conclusion.tex
\section{Conclusion}\label{sec:Conclusion}

In this paper we present two unifying learning-augmented algorithmic frameworks: A switching-based simple framework for online concave packing problems, using any state-of-the-art online algorithm as a black-box, and a primal-dual framework for online convex covering problems based on the ideas of~\cite{grigorescu2022arxiv} and~\cite{azar2016online}. Our algorithm for online concave packing, \Cref{alg:simple}, is $\frac{1}{1-\lambda}$-consistent, $\frac{\alpha}{\lambda}$-robust, and $(2-\lambda)\beta$-feasible, when given an $\alpha$-competitive $\beta$-feasible classical online algorithm as a black-box subroutine; Our algorithm for online convex covering, \Cref{alg:ConvexCover}, is a $O(\frac{1}{1-\lambda})$-consistent, $O((p \log \frac{d}{\lambda})^p)$-robust primal-dual learning-augmented algorithm, that accepts fractional problem instances and advice. In the context of online covering problems with $\ell_q$-norm objectives, we show that \Cref{alg:ConvexCover} can be adapted with slight modifications to obtain a $O(\frac{1}{1-\lambda})$-consistent, $O(\log \frac{\kappa d}{\lambda})$-robust algorithm. 

One open question immediately raised by our results is whether there are other specific forms of objective functions that our PDLA framework can adapt to and obtain improved robustness ratio over the general case. As illustrated by \Cref{alg:lqCover} and \Cref{thm:lqmain} in \Cref{sec:Lq}, our framework is general-purpose and flexible, and can easily specialize for $\ell_q$-norm objectives to utilize their structural properties. It would be exciting to see other settings for which our framework, or the PDLA framework in general, can be applied to.

Another potential direction following ideas in the field of learning-augmented algorithms is to lessen or remove our dependence on the external accuracy of the confidence parameter $\lambda$, the hyper-parameter to the algorithm. Our consistency and robustness ratios are functions of $\lambda$, and rely on it being a somewhat accurate representation of the performance of the advice; If the confidence parameter is chosen inappropriately, or in extremal cases adversarially, our PDLA algorithm can perform arbitrarily badly. While there have been efforts to accurately and efficiently learn hyper-parameters for learning-augmented algorithms via online learning techniques, there are learning-augmented algorithms for specific structured problems that does not require any hyper-parameters, and can maintain robustness against adversarial advice oracles. Replicating such parameter-independent robustness, or proving impossibility results, for general problems such as online covering programs with linear or convex objectives will be an important contribution to the field of learning-augmented algorithms.

We also raised the conceptual question of characterizing online problems with advice that can utilize classical online algorithms as black-boxes to obtain optimal performance. We believe that while meta-questions like these are hard to formulate and study, they are vital to further understand the fundamental power of learning-augmentation, and would like to invite more researchers to explore these directions.

%% file: convexvariant.tex
\section{Proof of \Cref{cor:convexcomplex}} \label{app:ConvexComplex}

In this section, we provide a formal proof of \Cref{cor:convexcomplex}. For completeness, we restate the corollary here:

\ConvexComplex*

As with the proof of \Cref{thm:convexmain}, we prove the two clauses of \Cref{cor:convexcomplex} separately: that $f(\bar x) \leq O(1/(1-\lambda)) f(x') + R(1) \opt$, and that $f(\bar x) \leq R(\lambda) \opt$.

\paragraph{Robustness.} We show the following subclaims:
\begin{enumerate}
    \item $\bar x$ is feasible and monotone;
    \item $(\bar y, \mu)$ is feasible, and $\mu$ is monotone;
    \item The primal objective $P$ is at most $R(\lambda)$ times the dual objective $D$.
\end{enumerate}
which, together with weak duality, implies:
\[P = f(\bar x) \leq R(\lambda) D \leq R(\lambda) \opt\]
as desired.

For the two feasibility subclaims, \Cref{lem:convexfeasibility} holds by algorithm design, regardless of our choice of $\mu$ and $r$, the dual growth rate.

Towards the third subclaim bounding the ratio between the primal and dual objectives, we continue using the sandwich strategy, starting with the lower bound of $x_j$ at time $\tau$, corresponding to \Cref{lem:convexxjbound}:
\begin{lemma}\label{lem:convexcomplexxjbound}
    In context of \Cref{alg:ConvexCoverComplex}, for a variable $x_j$, let $T_j = \{i | a_{ij} > 0\}$ and let $S_j$ be any subset of $T_j$. Then for any $t \in T_j$ and $\tau_t \leq \tau \leq \tau_{t+1}$,
    \begin{align*}
        x_j^{(\tau)} \geq\; & \frac{\lambda}{\max_{i \in S_j} \{a_{ij}\} \cdot d} \cdot \left( \exp \left( \frac{\log (1 + \frac{2}{\lambda}d^2)}{\mu_j} \sum_{i \in S_j} a_{ij} y_i^{(\tau)} \right) - 1 \right)
    \end{align*}
    where $\tau_t$ denotes the value of $\tau$ at the arrival of the $t$-th primal constraint.
\end{lemma}

\begin{proof}
Note that 
\begin{align*}
\frac{\partial x_j}{\partial y_i} &= \frac{\partial x_j}{\partial \tau} \cdot \frac{\partial \tau}{\partial y_i} = \frac{a_{ij} (x_j + D_j)}{\nabla_j f(x)} \cdot \frac{\log (1 + \frac{2}{\lambda}d^2)}{\min_{\ell=1}^n \left\{ \frac{\nabla_\ell f(\delta \bar x)}{\nabla_\ell f(\bar x)} \right\}}\\
&\geq \log (1 + \frac{2}{\lambda}d^2) \frac{a_{ij} (x_j + D_j)}{\nabla_j f(\delta \bar x)}\\
&= \log (1 + \frac{2}{\lambda}d^2) \frac{a_{ij} (x_j + D_j)}{\mu_j}
\end{align*}
where the inequality follows from the monotonicity of $\nabla f(x)$ and
\begin{align*}
    \min_{\ell=1}^n \left\{ \frac{\nabla_\ell f(\delta \bar x)}{\nabla_\ell f(\bar x)} \right\} \cdot \nabla_j f(x) &\leq \frac{\nabla_j f(\delta \bar x)}{\nabla_j f(\bar x)} \cdot \nabla_j f(x)\\
    &\leq \frac{\nabla_j f(\delta \bar x)}{\nabla_j f(x)} \cdot \nabla_j f(x)\\
    &= \nabla_j f(\delta \bar x)
\end{align*}
where the last inequality is due to our assumption that the gradient is monotone, and that $x \leq \bar x$.

Solving this differential equation, we have
\[\frac{x_j^{(\tau)} + D_j^{(t)}}{x_j^{(\tau_t)} + D_j^{(t)}} \geq \exp \left( \frac{\log (1 + \frac{2}{\lambda}d^2)}{\mu_j} \cdot a_{tj} y_t^{(\tau)} \right)\]

Multiplying over all indices, where for notational convenience we set $\tau_{t+1} = \tau$,
\begin{align*}
\exp \left(\frac{\log (1 + \frac{2}{\lambda}d^2)}{\mu_j} \sum_{i\in S_j} a_{ij} y_i^{(\tau)} \right) &\leq \prod_{i \in S_j} \frac{x_j^{(\tau_{i+1})} + D_j^{(i)}}{x_j^{(\tau_i)} + D_j^{(i)}} \\
&\leq \prod_{i \in S_j} \frac{x_j^{(\tau_{i+1})} + \frac{\lambda}{\max_{i \in S_j} a_{ij} d}}{x_j^{(\tau_i)} + \frac{\lambda}{\max_{i \in S_j} a_{ij} d}} \qquad \text{ by \Cref{claim:telescope} and \Cref{claim:Dj}}\\
&\leq \prod_{i \in T_j} \frac{x_j^{(\tau_{i+1})} + \frac{\lambda}{\max_{i \in S_j} a_{ij} d}}{x_j^{(\tau_i)} + \frac{\lambda}{\max_{i \in S_j} a_{ij} d}} \qquad \text{ since $x_j$ is monotone}\\
&= \frac{x_j^{(\tau)} + \frac{\lambda}{\max_{i \in S_j} a_{ij} d}}{\frac{\lambda}{\max_{i \in S_j} a_{ij} d}} \qquad \text{ by a telescoping argument over all indices}
\end{align*}

Reorganizing the terms yields the desired bound on $x_j^{(\tau)}$.
\end{proof}

With the lower bound of \Cref{lem:convexcomplexxjbound} and the trivial upper bound of $1/a_{tj}$, we continue to proving the third subclaim bounding the ratio between the primal and dual objectives, analogous to \Cref{lem:convexrobust}:
\begin{lemma}\label{lem:convexcomplexrobust}
Let $p := \sup_{x \geq \mathbf{0}} \frac{\langle x, \nabla f(x) \rangle}{f(x)}$. Then
\begin{align*}
    P &= f(\bar x) \leq R(\lambda) D
\end{align*}
where
\begin{align*}
    \frac{1}{R(\lambda)} = \max_{\delta > 0} \min_z &\left( \frac{\min_\ell \left\{ \frac{\nabla_\ell f(\delta z)}{\nabla_\ell f(z)} \right\}}{4 \log (1 + \frac{2}{\lambda} d^2)}  - \frac{ \delta z^T \nabla f(\delta z) - f(\delta z)}{f(z)} \right).
\end{align*}
\end{lemma}

\begin{proof}
We assume that $x'$ is feasible, i.e., $A_t x' \geq 1$. A similar analysis can be obtained for the case of $A_t x' < 1$ by setting $\lambda = 1$.

Consider the update when primal constraint $t$ arrives and $\tau$ is the current time. Let $U(\tau)$ denote the set of tight dual constraints at time $\tau$, i.e., for every $j \in U(\tau)$ we have $a_{tj} > 0$ and $\sum_{i=1}^t a_{ij} y_i^{(\tau)} = \nabla_j f(\delta \bar x) = \mu_j$. $|U(\tau)| \leq d$, and define for every $j$ the set $S_j := \{i | a_{ij} > 0, y_i^{(\tau)} > 0\}$. 
Clearly $\sum_{i \in S_j} a_{ij} y_i^{(\tau)} = \sum_{i=1}^t a_{ij} y_i^{(\tau)} = \nabla_j f(\delta \bar x) = \mu_j$, and $\sum_j a_{tj} x_j^{(\tau)} < 1$, thus by \Cref{lem:convexxjbound}, we have
\begin{align*}
    \frac{1}{a_{tj}} > x_j^{(\tau)} \geq\;& \frac{\lambda}{\max_{i \in S_j} \{a_{ij}\} \cdot d} \cdot \left( \exp \left( \log (1 + \frac{2}{\lambda}d^2) \right) - 1 \right)
\end{align*}

Rearranging, we have
\[\frac{a_{tj}}{a_{m^*_j j}} = \frac{a_{tj}}{\max_{i \in S_j} a_{ij}} \leq \frac{1}{2d}\]
Thus, we can bound the rate of change of $\sum_{i=1}^t y_i$ at time $\tau$:
\begin{align*}
    \frac{\partial (\sum_{i=1}^t y_i)}{\partial \tau} &\geq r - \sum_{j \in U(\tau)} \frac{a_{tj}}{a_{m^*_j j}} \cdot r \geq r \left( 1 - \sum_{j \in U(\tau)} \frac{1}{2d} \right) \geq \frac{1}{2} r
\end{align*}
where the last inequality follows from $|U(\tau)| \leq d$.

The rate of increase of the primal is unaffected by the changes we made in \Cref{alg:ConvexCoverComplex}. Thus, we can similarly bound the rate of change between the primal and the dual as before:
\begin{align*}
    \frac{\partial (\sum_{i=1}^t y_i^{(\tau)})}{\partial f(x^{(\tau)})} &= \frac{\partial (\sum_{i=1}^t y_i^{(\tau)})}{\partial \tau} \cdot \frac{\partial \tau}{\partial f(x^{(\tau)})}\\
    &\geq \frac{1}{2} r \cdot \frac{1}{2} = \frac{1}{4} r = \frac{\min_\ell \left\{ \frac{\nabla_\ell f(\delta \bar x)}{\nabla_\ell f(\bar x)} \right\}}{4 \log (1 + \frac{2}{\lambda} d^2)}
\end{align*}

Taking integral over $\tau$, we have
\begin{equation*}
    \sum_{i=1}^m \bar y_i \geq \frac{\min_\ell \left\{ \frac{\nabla_\ell f(\delta \bar x)}{\nabla_\ell f(\bar x)} \right\}}{4 \log (1 + \frac{2}{\lambda} d^2)} \cdot f(\bar x).
\end{equation*}

By definition of the dual objective $D$, we have
\begin{align*}
    D &= \sum_{i=1}^m \bar y_i - f^*(\mu)\\
    &\geq \frac{\min_\ell \left\{ \frac{\nabla_\ell f(\delta \bar x)}{\nabla_\ell f(\bar x)} \right\}}{4 \log (1 + \frac{2}{\lambda} d^2)} \cdot f(\bar x) - f^*(\nabla f(\delta \bar x))\\
    &= \left( \frac{\min_\ell \left\{ \frac{\nabla_\ell f(\delta \bar x)}{\nabla_\ell f(\bar x)} \right\}}{4 \log (1 + \frac{2}{\lambda} d^2)} - \frac{\delta \bar x^T \nabla f(\delta \bar x) - f(\delta \bar x)}{f(\bar x)} \right) f(\bar x)\\
    & \geq \min_z \left( \frac{\min_\ell \left\{ \frac{\nabla_\ell f(\delta z)}{\nabla_\ell f(z)} \right\}}{4 \log (1 + \frac{2}{\lambda} d^2)}  - \frac{ \delta z^T \nabla f(\delta z) - f(\delta z)}{f(z)} \right) \cdot P
\end{align*}
where the equality on the third line is from (2) of \Cref{fact:boundedgrowth}. 
Our desired result follows from maximizing this ratio and choosing $\delta$ to be the maximizer.
\end{proof}

\paragraph{Consistency.} We note that \Cref{lem:convexconsistency} is unaffected by our changes in \Cref{alg:ConvexCoverComplex}, since the primal growth rate is unaltered, and thus it holds for \Cref{alg:ConvexCoverComplex} as well.